\documentclass[USenglish,anonymous]{article}
\usepackage{amsthm}
\usepackage{amsmath}
\usepackage{amssymb}
\usepackage{graphicx}
\usepackage[utf8]{inputenc}
\usepackage[T1]{fontenc}
\usepackage{url}
\usepackage{hyperref}
\usepackage{subcaption}

\newtheorem{theorem}{Theorem}

\newtheorem{example}[theorem]{Example}

\def\bbN{\mathbb{N}}

\def\bbH{\mathbb{H}}
\def\ra{\rightarrow}

\def\bbE{{\mathbb E}}
\def\bbN{{\mathbb N}}
\def\bbH{{\mathbb H}}

\def\bbQ{{\mathbb Q}}
\def\bbZ{{\mathbb Z}}
\def\bbX{{\mathbb X}}
\def\ra{\rightarrow}

\def\ext#1{#1}

\title{Generating Tree Structures for Hyperbolic Tessellations}

\author {
    Dorota Celińska-Kopczyńska, Eryk Kopczyński \\
    { \small Institute of Informatics, University of Warsaw}
}

\begin{document}


\maketitle

\begin{abstract}
We show an efficient algorithm for generating geodesic regular tree structures for periodic hyperbolic and
Euclidean tessellations, and experimentally verify its performance on tessellations.
\end{abstract}

Hyperbolic geometry is characterized by tree-like, exponential growth: in the hyperbolic plane $\bbH^2$,
a circle of radius $r$ has area and circumference exponential in $r$. This property has recently found
applications in data visualization \cite{lampingrao,munzner},
social network and data analysis \cite{papa,bogu_internet}.
However, numerical precision errors are a serious issue in computational hyperbolic geometry. 
Since the area and circumference of a circle of radius $r$ is exponential, any coordinate system
based on a fixed tuple of coordinates will fail to distinguish two points in distance $1$ already
in distance proportional to the number of bits used to represent the coordinates.
One way to handle this issue is to use tessellations. Hyperbolic tessellations can be
generated combinatorically, avoiding the numerical issues altogether, and have many applications
by themselves, e.g., in art \cite{coxe41}, game design \cite{hyperrogue}, data visualization
\cite{hrviz}. Of course, tessellations are also a beautiful area of pure mathematics.

While generating combinatorial structures for simple hyperbolic tessellations, such as the order-3 heptagonal
tessellation \cite{trigrid} is relatively straightforward, finding them for arbitrary tessellations
is a challenging problem. In this paper, we present an algorithm for generating geodesic regular
tree structures (GRTS) for a given tessellation description, which in turn can be used to efficiently
generate the graph.
The core of our algorithm is conceptually similar to Angluin's algorithm of learning regular languages \cite{angluin}.
The core algorithm runs in polynomial time (under specific assumptions) when given the correct graph
of the tessellation; however, giving the correct structure is challenging when we do not know the
tree structure yet, and we can only provide an approximation. We present methods of constructing
an efficient approximation, that is, one that causes our whole algorithm to run efficiently.
Based on the experimental evaluation, we conjecture that our algorithm works in time polynomial
in the number of tile types for a fixed bound on the degrees of tiles and vertices.

{\bf Structure of the paper.} In Section \ref{sec:periodic} we define the periodic tessellations we
are working with. In Section \ref{sec:trees} we describe the tree structure, and explain how to
use such a tree structure to generate a periodic tessellation -- such tree structures are used both
in our algorithm and in earlier methods. In Section \ref{sec:previous} we discuss the limitations
of earlier methods. Our algorithm is described in Section \ref{sec:generating}. In \ref{sec:app}
we provide some applications, and in \ref{sec:experimental} we provide the experimental results
of our algorithm.

\section{Periodic tessellations}\label{sec:periodic}
We will work with periodic tessellations of $\bbX$, which can be either the Euclidean plane $\bbE^2$
or the hyperbolic plane $\bbH^2$. While our main focus is $\bbH^2$ where tree structures are
essential, $\bbE^2$ is useful for constructing counterexamples. While concrete periodic tessellations
in $\bbX$ are important for visualization and presentation purposes, we can albo consider abstract 
descriptions which describe their combinatorial structure as periodic planar graphs.

A periodic tessellation $C$ has a finite set of polygonal tile types (orbits) $T$. 
Each tile type $t\in T$ has a shape $S_t$, which is a polygon with $s_t$-fold rotational symmetry
and $N_t=n_t \cdot s_t$ edges, indexed clockwise from $0$ to $n_ts_t-1$. For convenience, we will consider
this ordering cyclic, so e.g. the edge $N_t+3$ is the same as edge 3. 
Every tile $c$ in the tessellation
$C$ is assigned $t(c) \in T$ and a clockwise enumeration of its edges. There is an isometry 
between the tile $c$ and $S_t$, which takes $i$-th edge of $c$ to $i$-th edge of $S_t$.
Furthermore, if $t(c_1) = t(c_2)$ and $r \in \bbZ$, there is an orientation-preserving isometry
$f$ of $\bbX$ which maps $c_1$ to $c_2$, $i$-th edge of $c_1$ to $(i+rn_t)$-th edge of $c_2$,
every other tile $c \in C$ to some tile $f(c) \in C$ such that $t(c) = t(f(c))$.

Because of periodicity, for every $i \in \bbZ$ and $t \in T$ there exists $i' \in \bbZ$ and
$t' \in T$ such that edge $i+rn_t$ (for all $r\in\bbZ$) of $t$ always connects to edge
$i'+r'n_{t'}$ of $t'$ for some $r\ \in \bbZ$. We say that $(t,i)$ connects to $(t',i')$,
$(t',i')=c(t,i)$. The shapes $S_t$ and the connection rules for every $t\in T$ determine the periodic tessellation
$C$ uniquely (up to isometry).

The $i$-th vertex of a tile $c \in C$ is the one incident to edges $i-1$ and $i$. The \emph{valence}
of a vertex is the number of tiles in $C$ which are adjacent to it. Because of periodicity,
the valence $v(t,i)$ is uniquely determined by $t(c)$ and $i$.

\begin{figure}[ht]
\begin{center}
\includegraphics[width=.32\linewidth]{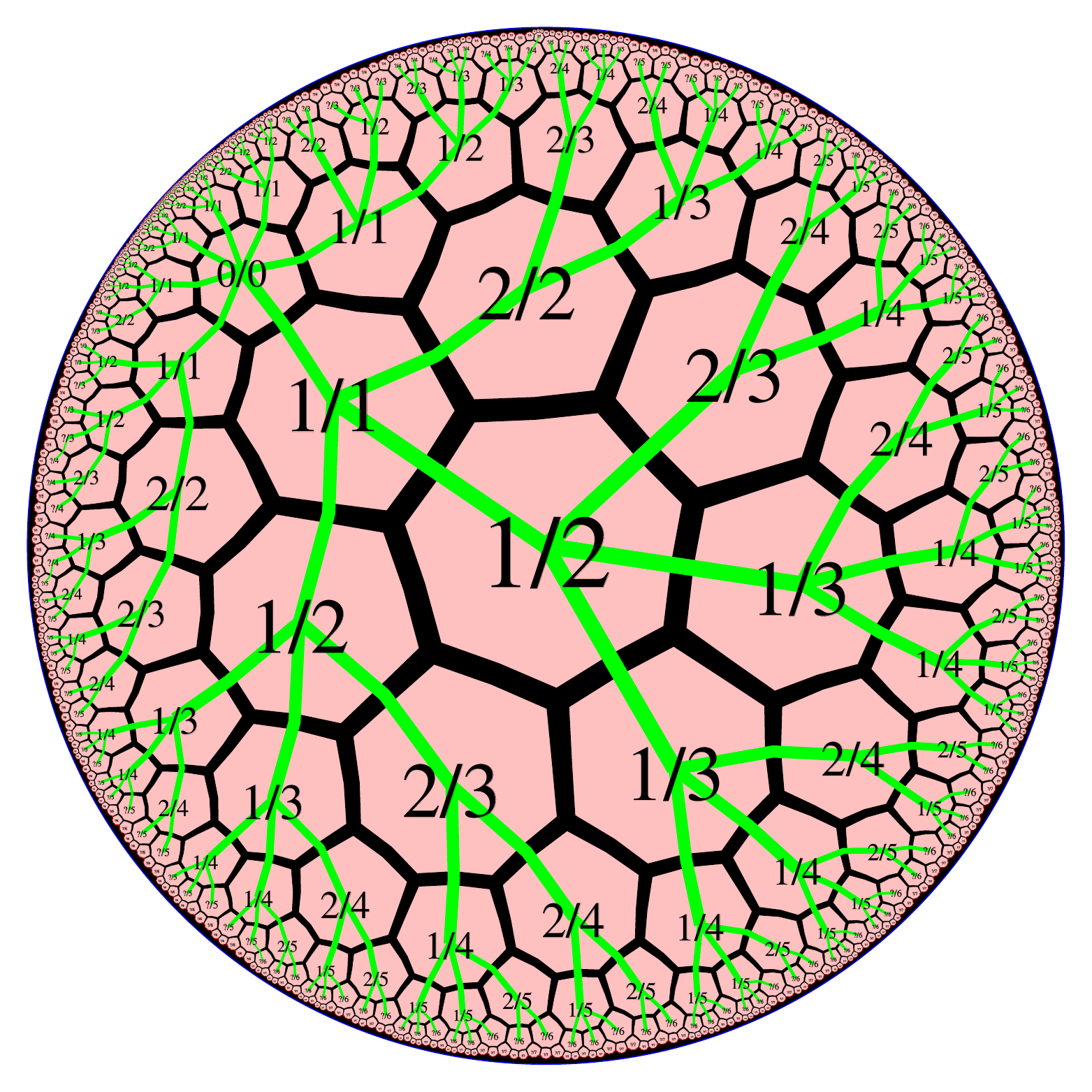}
\includegraphics[width=.32\linewidth]{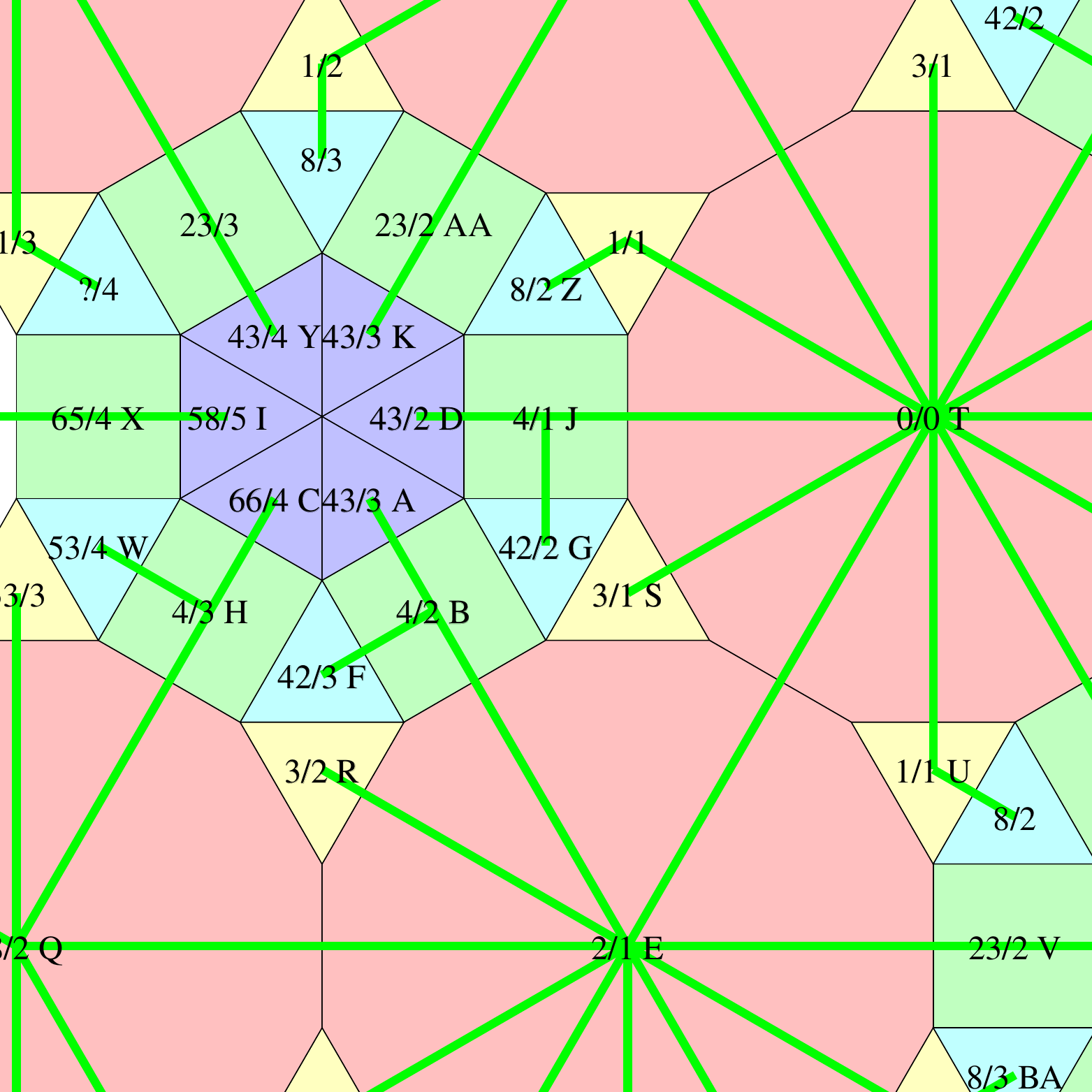}
\includegraphics[width=.32\linewidth]{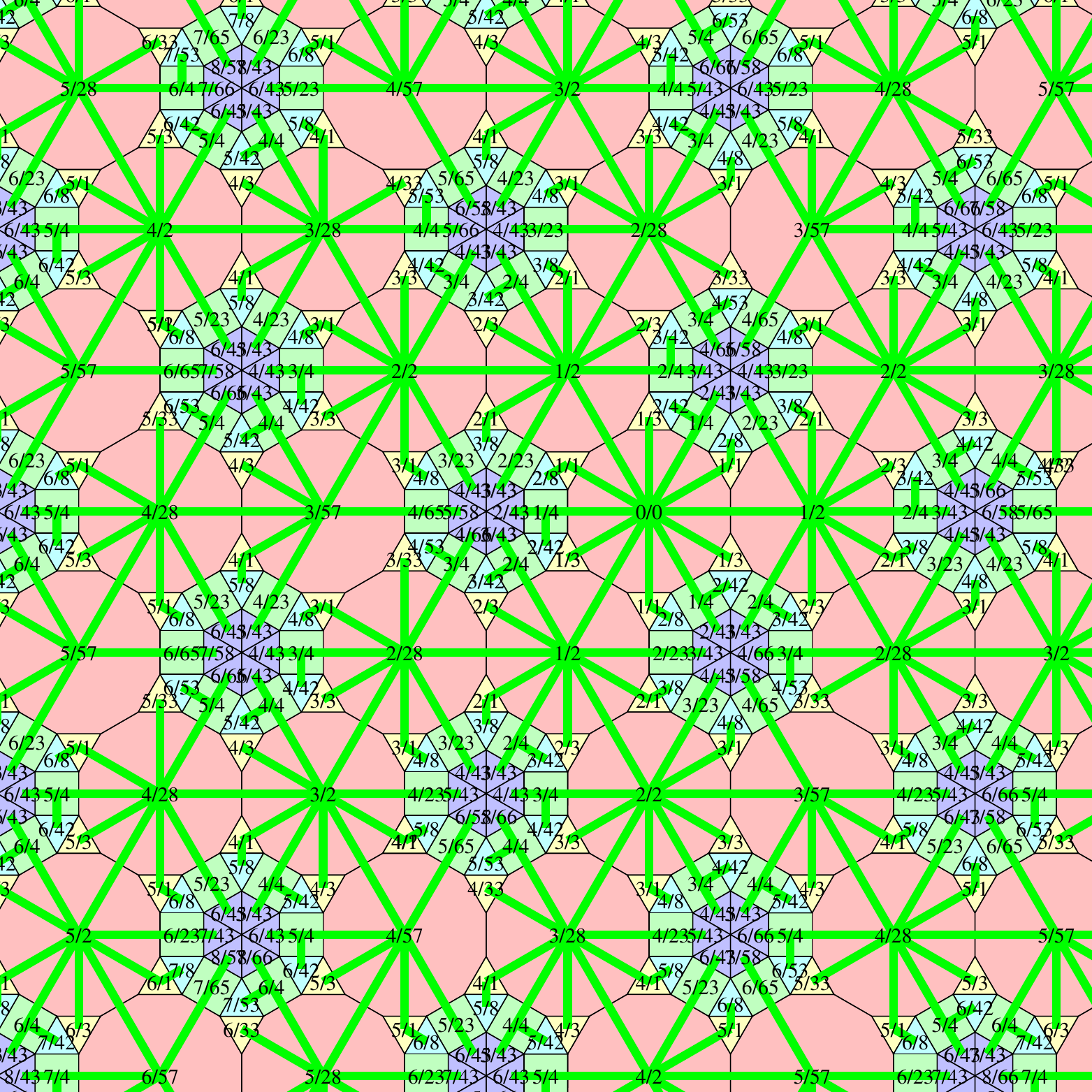}
\end{center}
\caption{From left to right: (a) order 3 heptagonal tiling, (bc) A complex Euclidean tessellation, zoomed in and out.
The first number is the state in the tree structure, and the second number is the distance from the root tile.
The letters are labels used in the description of Example \ref{complextes}.
\label{examplefig}}
\end{figure}

An {\bf abstract tessellation description} (ATD) is $D = (T, n, s, v, c)$, where $T$ is a finite set, $n:T\ra\bbN$
and $s:T\ra\bbN$ determine the sizes and symmetries of tiles, $v: T' \ra \bbN$ determines the valences
of vertices, and $c : T' \ra T'$ determine the connection rules. By $T'$ here we denote the set of edge types
$\{t,i\}: 0 \leq i < n_t\}$. An abstract tessellation description is {\bf consistent} iff $c(c(t'))=t''$ and
for every $t' \in T'$, we have $v(t')=v(n(t'))$ and $n^{v(t')}(t')=t'$, where $n$ is the next edge around the
vertex, i.e., if $c(t,i)=(t',i')$, $n(t,i)=(t',i'+1 \bmod n_{t'})$.

A consistent ATD uniquely determines
the structure of the tessellation $C$. Furthermore, most abstract combinatorial structures
can be realized as concrete periodic tessellations of $\bbH^2$. Indeed, such a $C$ can be realized as a topological surface
$M$ by gluing $N_{t(c)}$-gons for all tiles c of $C$ according to the connection rules. By identifying the points of $M$
according to its symmetries, we obtain a quotient orbifold. We can compute the Euler characteristics $\chi\in\bbQ$ of this orbifold
based on $D$. Except for a few \emph{bad} orbifolds, such an orbifold can be realized as a quotient (by some discrete group of isometries)
of a hyperbolic tessellation if $\chi<0$, Euclidean tessellation if $\chi=0$, spherical tessellation if $\chi>0$ (Theorem 2.4, \cite{scottgeoms}).

\section{Tree structure}\label{sec:trees}
In this section we describe the main tool we will be using to generate tessellations.

A {\bf regular tree structure} (RTS) is $(Q, t:Q \ra T, e: Q \times \bbZ \ra Q \cup \{L, R, P\})$.
The tree structure generates a tessellation $C$ as follows. 
Every tile $c \in C$ is assigned a state $q(c)$ such that $t(q(c)) = t(c)$. 
The transition rule $e(q,i)$ denotes every edge of $c$ as a parent ($P$), left ($L$),
right ($R$), or child ($q \in Q$) connection. For convenience we index transition
rules with all integers $i$, but we have $c(q,i) = c(q,i+N_{t(q)})$, so
all information is finite. Every tile has at most one parent connection, and
if $c(q,i)=q'$, $(t(q),i)$ must connect to $(t(q'),p(q'))$, where $p(q')$ is the index
of the parent connection in $q'$.

\begin{example} The simplest tessellation is $\{7,3\}$, the order-3 heptagonal tessellation (Figure \ref{examplefig}a).
We have only one tile type $t$ which is a heptagon, with $s_t=7$ and $n_t=1$.
For every $i, i'$, $(t,i)$ connects to $(t,i')$. This tessellation has tree states
$Q=\{0,1,2\}$. The connection rules for each $q \in Q$ and indices $i=0..6$
are: $e(0)=(1,1,1,1,1,1,1)$, $e(1) = (P,L,2,1,1,R,R)$, $e(2) = (P,L,L,2,1,R,R)$.
\end{example}

A RTS lets us generate the tessellation $C$
combinatorially. The tessellation will be generated lazily;
let $G$ be the set of tiles generated so far.
For every $g\in G$, we keep the following information: 
its state $q(c)$, and for every $i=0 \ldots N_t-1$, its connection
$e(g,i$), which is either a pointer to
another tile $g'\in G$ and an index $i'$ (meaning that edge $i$ of $g$
connects to edge $i'$ of $g'$) or NULL (meaning that we do not know
this connection yet).

A \emph{treewalker} is a structure consisting of
a pointer to a cell $g \in G$, and edge index $i$. Intuitively,
a treewalker is currently in $g$ and facing the direction given by $i$.
A walker can step forward (compute $e(g,i)$ if it is currently NULL, 
and then move to $e(g,i)$), rotate counterclockwise to $(c,i-1)$ or
clockwise to $(c,i+1)$.

We start with $G=\{g_0\}$, where $q(g_0) = q_0 \in Q$ such that
$q_0$ has no parent. We will generate the other tiles as needed.
Suppose we want to check what tile is adjacent to an already generated
tile $g$, at edge $i$. If $e(g,i)$ is known, nothing needs to be
done. Otherwise, let $x = e(q(g),i)$.

(1) If $x\in Q$, we generate a new tile of
type $t(x)$ and connect to edge $p(x)$ of the new tile. (Connecting
means setting $e$ on both sides.) This way, we
generate a tree of tiles.

(2) If $x = R$, we construct the connection by going counterclockwise
around the $i+1$-th vertex of $c$. Take the treewalker $(g,i+1)$,
and $v(t(q(g(i))),i)-1$ times step forward and rotate clockwise.
(Stepping forward may require calling our algorithm recursively if
the given edge is not yet generated). After $v(t(q(g(i))),i)-1$
iterations our treewalker is in state $(g', i')$, where 
$e(q(g'),i') = L$. We connect $(g,i)$ to $(g',i')$.

(3) The situation $x = L$ is symmetrical.

(4) By construction, the situation $x = P$ may not happen: we have
started in $q_0$ which has no parent connection, and for every
tile $g$ generated later, its parent connection starts
connected to a previously generated vertex.

\begin{example}\label{complextes} The tessellation in Figure \ref{examplefig}(bc) is a more complex tessellation
with 5 tile types. The algorithm described in the sequel generates a
RTS with $|Q|=69$.

Rule (1) lets us connect all the tree edges (colored green). 
For clarity, assume that the parent of every tile is indexed by 0.
Suppose we now want to find neighbor 1 of tile $A$ in state 43.
The transition rule $e(43,1)=L$, so we should go clockwise around
the vertex $ABFHC$. The first move ($AB$) is a parent edge, and the second move 
($BF$) is a child edge, $e(4,1)=42$ (rule 1). For the next move $FH$ we see
$e(42,2)=L$, so the connection algorithm will have to be called recursively.
(Note that $42$ and $43$ are different states because $e(43,2)=R$.)
Then, $e(4,2)=66$ so we have a child edge to $C$. We connect the 2nd edge of $C$
to the 1st edge of $A$.

In the recursive call, we go around the vertex $FRQH$, in turn recursively
going around the vertices $REQ$ and $FBER$.
\end{example}

If the tree structure is correct, the tree obtained by connecting
every tile according to rules (1) and (4) will generate a tree,
and every tile $c \in C$ of the actual tessellation will appear
exactly once in $G$. Furthermore, if the assignment of $L$ and $R$
is correct, rules (2) and (3) will let us combinatorially determine
all the connections, and the resulting $G$ is essentially a copy
of $C$, constructed combinatorially.

We also want our RTS to be \emph{geodesic} (GRTS),
i.e., if tile $g'$ is a child of $g$, we have $\delta(g') = \delta(g)+1$,
where $\delta(g)$ is the length of the shortest path from $g_0$ to $g$.
That is, the depth of a tile $g$ equals its distance from the root tile $g_0$.

\section{Previous methods}\label{sec:previous}
A na\"ive method of generating hyperbolic tessellations is to compute the coordinates of every
tile. If we get the same coordinates as a previously existing tile, we know that we connect to it.
More precisely, we represent every point in $\bbH^2$ using three coordinates in the Minkowski hyperboloid model \cite{cannon},
and its isometries as $3 \times 3$ matrices. For every tile $g \in G$ we compute the matrix of the isometry which takes
the shape $S_{t(g)}$ to $g$. When creating a new connection to a tile of type $t$ at the location given by isometry $T$, we see if
there is already a tile $g'$ of that type at locating given by isometry $T \cdot R_{2\pi i/n_t}$ for $i \in \bbZ$, where $R_\alpha$
is the matrix of rotation by angle $\alpha$; if yes, we know that we connect to $g'$, if no, we create a new $g'$ there.

This method is not suitable because of the precision issues inherent to computations in hyperbolic geometry.
To explain the issue,
let $d(A,v,r)$ be the point $r$ units from the given point $A$ in direction $v$. Suppose that, due to the numerical precision issues, we represent the direction
$v$ as $v'$, where $|v-v'| \leq \epsilon$. However, a circle of radius $r$ has circumference of $2\pi \sinh(r)$, which is exponential in $r$,
and thus the distance between $d(A,v,r)$ and $d(A,v',r)$ can be of order $\epsilon \times e^r$ \cite{hamilton2021nogo}.
In general, the number of tiles in $r$ steps from the center is exponential in $r$, so any representation using a fixed number of bits will not be able to
discern between the coordinates of two tiles if $r$ is large enough, on the order of the number of bits.
For points faraway from the center of the model, we also get representation issues when points are too close to the boundary in the Poincar\'e model,
or too large in the Minkowski hyperboloid model \cite{achilles}.
Thus, while
this method works without significant problems in Euclidean geometry, and also can be used when we are only interested in generating small
neighborhoods of some point in $\bbH^2$, it behaves very badly when we want to generate tiles of a hyperbolic tessellation in a large radius.
Thus, better methods of working with hyperbolic tessellations involve representing them in a purely combinatorial way.

One such combinatorial representation is based on the theory of automatic groups \cite{wpigroups,levygen}. Consider a tessellation of the hyperbolic plane
with a \emph{Coxeter triangle} with angles $180^\circ/p$, $180^\circ/q$ and $90^\circ$. Let $a$, $b$, and $c$, respectively, be the reflections of this triangle
in its edges opposite of the angles. By considering all the possible sequences of reflections, we obtain a tessellation of the hyperbolic plane; this is the
same tessellation one can obtain by subdividing every tile of the $\{p,q\}$ tessellation into $2p$ Coxeter triangles. We can represent every tile as a sequence
of reflections $a$, $b$, $c$ which take some origin Coxeter triangle $t_0$ to it; thus, the set of all triangles is a group $G$ generated by $a$, $b$ and $c$, where $a^2=b^2=c^2=(ab)^2=(ac)^q=(bc)^p=e$. 

Let $\phi: \{a,b,c\}^* \ra G$ be the group homomorphism which maps a sequence of symbols $a$, $b$ and $c$ to the group element it represents.
A word $w \in \{a,b,c\}^*$ is a geodesic if there is no $w_2 \in \{a,b,c\}^*$ such that $\phi(w_2) = \phi(w)$ and $|w_2| < |w|$.
The group $G$ is hyperbolic and thus it is
strongly geodesically automatic \cite{wpigroups}, which means that the set of all geodesic words is regular, there is a transducer $T_\epsilon$ which take two words $v$, $w$ and accept them 
iff $\phi(v)=\phi(w)$, and for every $x \in a, b, c$ there exists a transducer $T_x$ which takes two words $v$, $w$ and accepts them iff such that $\phi(v)x=\phi(w)$.

Let $\preccurlyeq$ be an order on words in $\{a,b,c\}^*$ that can be recognized by a transducer and such that $u\preccurlyeq v$ implies $uw \preccurlyeq vw$. 
(A transducer here is a deterministic finite automaton over $\{a,b,c,\$\}^2$ which reads two words in parallel; $\$$ is used to pad the word which ends earlier.)
For every element $g$,
let $s(g)$ be the word in $\{a,b,c\}^*$ such that $\phi(s(g))=g$ and $s(g)$ is the smallest according to the order $\preccurlyeq$. By the closure properties of regular languages,
$s(G)$ is also a regular language. We can find a deterministic finite automaton (DFA) recognizing
this regular language, which is essentially a GRTS that can be used to
generate the tessellation. The Knuth-Bendix completion algorithm \cite{useofknuth} can be used to easily find this DFA based on our relators $a^2=b^2=c^2=(ab)^2=(ac)^q=(bc)^p=e$.

This method works for for tessellations by Coxeter triangles and also their higher-dimensional analogs. Since these tessellations are strongly related to regular
tessellations, we can also build regular tessellations on top of them (possibly not geodesic). However, the process is quite complex, and it is not clear how to
generalize the Coxeter triangle approach to more complex tessellations.

The theory of automatic groups can be adapted to any periodic hyperbolic tessellation, by considering symbols corresponding to all possible moves through the edges.
We will again describe the possible paths in our tessellations as sequences of symbols. Imagine a walker at some tile of type $t$ facing edge $j$ such that
$n_t | j$. A symbol $(t,i)$ for every $i \in \{0, \ldots, N_t-1\}$ rotates us by $i$ edges to face edge $(j+i) \bmod N_t$, steps forward 
to tile of type $t'$ and edge number $i'$, and rotate back by $i' \bmod n_{t'}$ so that the obtained facing $j'$ again satisfies $n_{t'} | j'$. 
The set of all sequences of symbols describing a valid path is a hyperbolic groupoid, and most of the theory above is still sound for hyperbolic groupoids, in particular,
the set of shortlex smallest representations is again a regular language. However, outside of the simplest cases such as Coxeter triangles, the Knuth-Bendix completion
algorithm is not likely to terminate.

In \cite{trigrid} a method is presented for generating the tessellations obtained using the Goldberg-Coxeter construction on a regular tessellation $\{p,3\}$.
This method uses a GRTS, and is similar to the methods used in \cite{hyperrogue,margenstern_heptagrid,margenstern_pentagrid}.
It relies on the fact that, in such tessellations, the set of tiles in $d$ steps from the chosen root tile $g_0$ forms a cycle. A similar approach also works
for the regular tessellation $\{p,4\}$ (and tessellations obtained from it using the Goldberg-Coxeter construction). However, this assumption is not satisfied
even in very simple cases, for example, the face-transitive tilings with face configuration V5.8.8 or V14.14.3 \cite{trigrid}.

In \cite{colorbook} GRTS are used to determine the coordination sequences for Euclidean tessellations. This is more of an informal method
rather than an algorithm working on all tessellations.

\section{Generating a tree structure}\label{sec:generating}

In this section we describe our algorithm for generating a tree structure.
The input is an ATD $D$. The output will be a GRTS generating a tessellation consistent with $D$.

In the first two subsections we present the general idea of our algorithm. In these subsections,
we assume that we have access to our tessellation, i.e.,
we have a set $G$ which correctly represents the set of all tiles, and for each tile $g \in G$
we know its type $t(g) \in T$, connections $e(g,i) \in G \times \bbN$, and $\delta(g)$,
the distance from the root.

In practice, we do not have that information -- while it can be easily obtained if we already
have a GRTS, during the run of our RTS-constructing algorithm we can use only an approximation.
This approximation is detailed in the later subsections.

\begin{figure}[ht]
\begin{center}
\includegraphics[width=.6\linewidth]{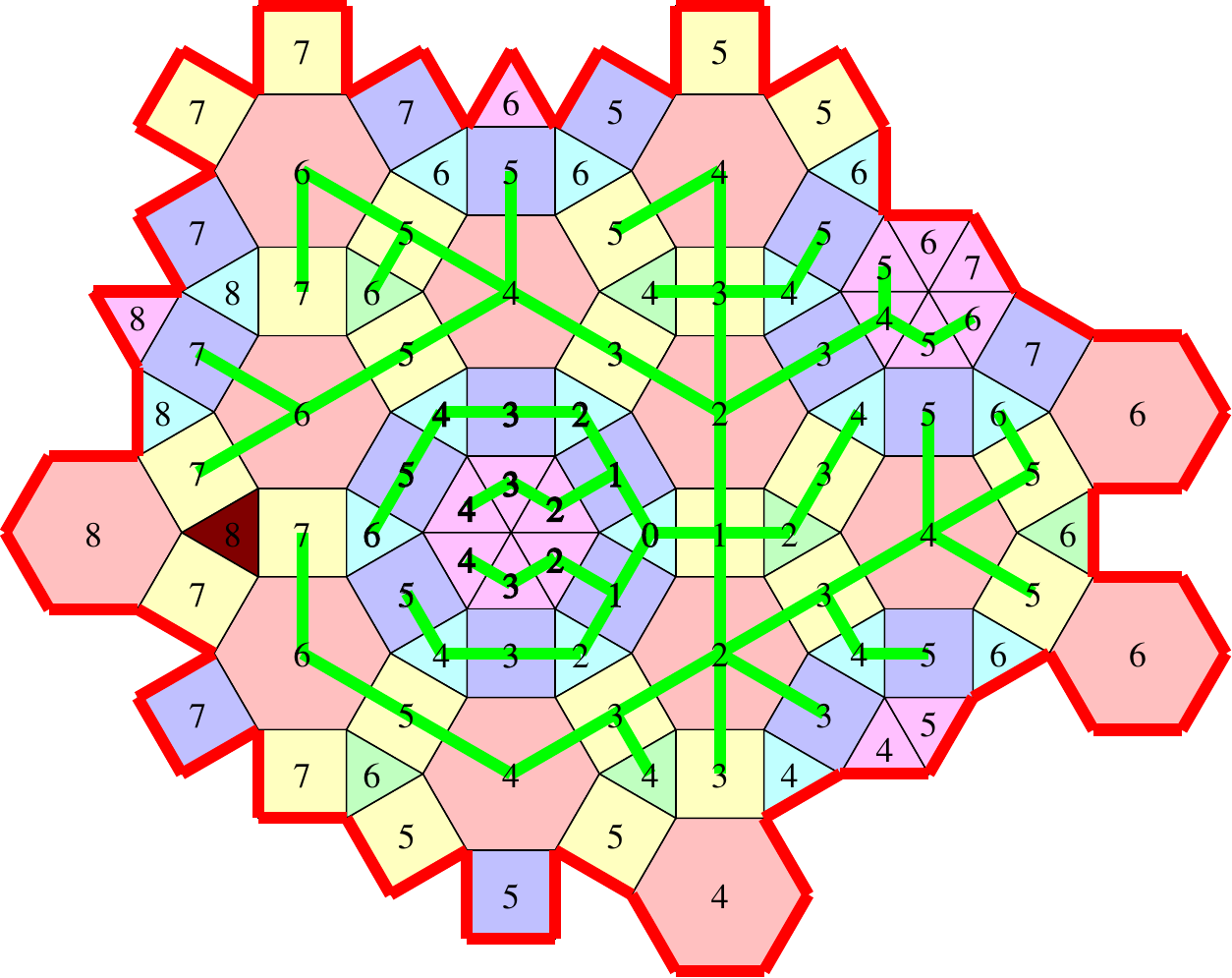}
\end{center}
\caption{This tessellation has threefold symmetry at the dark tile. Surprisingly, every
direction gets us closer to the origin. Snapshot at the time of checking the parent direction for the dark tile.\label{symtriangle}}
\end{figure}

\subsection{Determining the parent and sides}\label{sec:parent}

In a GRTS, every $g \neq g_0$ will have a parent
$p(g)$. We know $\delta(p(g)) = \delta(g)-1$. If there is only one neighbor with this
property, we can assume that it will be the parent (as long as the assumed values of
$\delta$ are correct); otherwise, we need to use some tiebreaker rule to pick between the possible
candidates. Once we pick the parent for every tile, we can consider two tiles equivalent if
their subtrees have the same shape (Myhill-Nerode equivalence). The minimal GRTS producing that tree
will have a state for every equivalence class. Picking a correct tiebreaker rule is crucial -- an
incorrect rule potentially corresponds to a tree with infinitely many differently shaped subtrees. 
As an example, consider the Euclidean square
grid with tiles indexed by their pair of coordinates ($\bbZ^2$), and the tiebreaker rule which chooses that the parent of $(x,y)$ ($x,y>0$) is $(x-1,y)$
if $x<y$, and $(x,y-1)$ if $x\geq y$. With such a tiebreaker rule, $(1,y)$ would have a different
finite number of descendants for every $y>2$, and thus they all would have to be in different
states.

Let $I$ be the set of parent candidates. We use the following tiebreakers for a tile of type $t$:
\begin{itemize}
\item for every $i \in I$, compute $i \bmod n_t$ -- pick the smallest one.
\item the rule above may not yield a winner if parent candidates surround $g$
(Figure \ref{symtriangle}).
In that rare case, construct the shortest paths generated by every candidate parent as
sequences of integers (by how many edges should the treewalker rotate at every point in the
path), and pick the lexicographically earliest one.
\end{itemize}

\begin{theorem}
For hyperbolic tessellations, the rules above generate a GRTS (with finitely many states).
\end{theorem}

This follows from the theory of automatic groups (see Section \ref{sec:previous}) for an appropriate
order $\preccurlyeq$. Since our main focus is on hyperbolic tessellations, we have not proven this
in the Euclidean case, although our extensive testing on Euclidean periodic tessellations suggests this is also
true.

We also need to classify non-tree edges into \emph{left} edges and \emph{right} edges, corresponding to 
$L$ and $R$ transitions of the tree structure. Non-tree edges form a forest, where every
connected component is infinite (if it was finite, we could go around the wall, and so the tree structure
given by parent directions would not be a tree) but it does not split $\bbX$ into two regions (it has only
one infinite branch). We call these connected components \emph{walls} (imagine a wall placed on every non-tree
edge). Remember that a treewalker represents being at a given tile and facing the given direction.
Start at $s$ and face edge $i$.  Try moving around the wall. If we eventually reach $(t,j)$ by moving right,
we know that $t$ is to the right; if we eventually reach $(s,i)$ by moving left, we know that $s$ is to the left.

We use the procedure {\sc GetSide}$(t,j)$ to determine the side of a non-tree edge $(t,j)$. This procedure
works by creating two walkers at $(t,j)$, one moving left and one moving right. The one which is closer to 
the root moves first, which guarantees that eventually one of them will reach $(s,i)$ and thus determine the side.
This procedure can be greatly optimized by caching. We cache the result for the next call of ${\sc GetSide}$;
also note that this will also determine the side for every wall edge touched on the way, so we cache these too. Also, 
if the left (right)-moving treewalker is on a known left (right) edge, it can immediately cross the wall.
This way, a call of {\sc GetSide} which takes $m$ moves generates $\Theta(m)$ new information.

\subsection{Determining the states and transition rules}\label{sec:dettrans}

Our algorithm of determining the GRTS is conceptually similar to Angluin's classical algorithm of learning regular
languages \cite{angluin}. We keep the list of tiles $S$ that represent distinct states, and a set of rules $E:G \ra Q$
which determine the state a given tile belongs to. We are looking for a pair of $(S,E)$ which is closed 
(for every $g\in S$, and every child $g'$ of $g$, $E(g') = E(g_1)$ for some $g_1 \in S$) and consistent
(if $E(g_1)=E(g_2)$ and $g_1'$ and $g_2'$ are matching children of $g_1$ and $g_2$ respectively then $E(g'_1) = E(g'_2)$).
In each iteration, if the pair $(S,E)$ is not closed, we add every such $g'$ to $S$. If it is not consistent,
we extend $E$ in such a way that $E(g_1) \neq E(g_2)$. Eventually, we obtain a pair $(S,E)$ which is both
closed and consistent. In this situation, we can use the information obtained to build a regular tree structure,
and check whether this regular tree structure correctly builds the tessellation $C$. If yes, we return this
structure. Otherwise, we add any find inconsistencies we find to the set $S$.

Initially $S$ contains only the root $g_0$, and the rules $(E)$ are initially as follows. Initially, to classify $g \in G$, take the walker $w=(g,p(g))$. 
For every $i \in N_{t_g}$, we rotate the walker $w$ by $i$ steps, to $w^i=(g,i')$. Let $v^i=(g',i'')$ be the
walker obtained by stepping $w_i$ forward. We classify the relationship of $g$ and $g'$ as as one of the 8 cases:
parent, child, or 6 possibilities of non-tree edges, corresponding to $\delta(g)-\delta(g') \in \{-1,0,1\}$
and whether $t$ is to the right or to the left of $s$, according to {\sc GetSide}$(w)$. Initially, $g_1$ and $g_2$
are classified as the same state iff $t(g_1)=t(g_2)$, $p(g_1) = p(g_2) \bmod n_{t(g)}$, and their neighbor classifications are equal.

Inconsistency happens when we have $E(g_1)=E(g_2)$ but for some $i$, $v_1^i$ and $v^2_i$ are children of $g_1$
and $g_2$, but $E(v^1_i) \neq E(v^2_i)$. To handle such inconsistencies, we allow the rules to similarly classify
edges of the descendants of $g$. We use an approach based on decision trees. The rules are based on a series
of queries of form ``consider the edge we asked about in the earlier $i$-th query and that turned out to be a child
edge, let $w=(g',p(g'))$ be this child; what is the classification of $w$ rotated by $j$ steps?'' (the 0-th
query is $g$ itself, and the first $N_{t(g)}$ queries are about the neighbors of $g$). Each node of the decision tree 
is characterized by $i$, $j$, and has branches corresponding to the 8 possible answers. In the case of inconsistency,
we find the point in the decision trees for $v^1_i$ and $v^2_i$ which caused $E$ to characterize them as different
states, and we extend the leaf corresponding $E(g_1)$ to ask the same questions.

Once we obtain a pair $(S,E)$ which is closed and consistent, generating a candidate GRTS is straightforward.
Now, we validate whether the obtained tree structure is correct. To this end, for every state $q \in Q$
obtained, we determine whether it is a \emph{dead branch} or a \emph{live branch} -- a branch is dead if
it has finitely many descendants, and live otherwise. 
We determine dead branches by applying the criterion ``a branch is dead if all its children are dead''
until all dead branches are known.

The tree structure is correct if it correctly generates the structure of all the walls (i.e., non-tree edges).
For every state $q$, we find out all pairs of its live children $(i_1, i_2)$ such that there is no other
live child nor a parent between $i_1$ and $i_2$. Every pair of this kind corresponds to a wall --
the tile of state $q$ is the tile at the bottom of the wall (i.e., one adjacent to the tile adjacent to
the wall with the smallest $\delta$). We need to check whether this wall correctly splits the subtree
starting at a tile of state $q$ into the left and right side.

The function to validate the wall is called {\sc ExamineBranch}.
Again, we use two treewalkers touching the wall, both starting at $(t,i_1)$, where $q(t)=q$, one always moving to
the left and one always moving to the right. If the two treewalkers touch the two sides of the same
edge, and the transition rule labels the non-tree edge as "right" for the left treewalker and as "left" 
for the right treewalker, we advance both of them. If the right treewalker touches a non-tree edge
labelled as "right", we push that edge to a stack, and move forward; when it touches one labelled
as "left", we first pop edges from the stack before consulting the other treewalker. We also use a
similar stack for the left treewalker (by symmetry).

We also need to know when to stop. When both treewalkers are advanced, we note the current state of
both, and if they are currently in dead branches, also all the dead ancestors. If we see the
same sequence of states twice (possibly during an examination of another triple $(q,i_1,i_2)$), 
we know that the sequence will repeat, and thus the wall correctly splits the subtree rooted at $q$.
Whenever one of the two treewalkers moves from a tile $g$ to its child $g'$, we check whether the
code of $g'$ is what was expected from the transition rules for $q(g)$. If no, we add $g$ to the list of
important tiles $L$, and we need to repeat the main loop of our algorithm.

\begin{theorem}\label{divine}
If $G=C$ and $\delta$ are correct, this algorithm will generate a GRTS for $C$ in time polynomial in the 
number of different subtree shapes (according to the chosen parent rule), degrees of tiles and vertices, and the size
of dead branches.
\end{theorem}

We conjecture that, for a fixed bound on the degrees of tiles and vertices, the number of states and the size
of dead branches are polynomial.

\subsection{Generating approximate tessellations without a GRTS}\label{sec:nogrts}
The problem with Theorem \ref{divine} is that satisfying its assumptions about $G$ and $\delta$
is not straightforward. Here we describe our method of dealing with the issue. 
Rather than using $G=C$, our algorithm will lazily generate (better and better) approximations of the tessellation $C$ and $\delta$ for its own use.

For every already generated tile $g \in G$ we know its
type $t(g) \in T$, and connections $e(g,i)$, which again can be either NULL or
$(g',i') \in G \times \bbZ$. We start with $G= \{g_0\}$, where $g_0$ is an
unconnected tile of arbitrary type.

Whenever a treewalker $(g,i)$ steps forward and $e(g,i)$ is $NULL$, we consult
the connection rules to learn that $(t(g),i)$ connects to $(t', i')$, create
a new tile $g'$ of type $t(g')=t'$, and connect $(g,i)$ to $(g',i')$. We also
perform the \emph{valence check}: how many of tiles around the $i$-th vertex of $g$
we already know. If we know $v(t(g),i)$ of them, we also connect the first and
last one of them in the obvious way. Symetrically, we also similarly check
the $(i+1)$-th vertex of $g$.

It may happen that an inconsistency (\emph{unification error}) is detected, i.e., we know more than $v(t(g),i)$
of them. This is because the generation rule above can effect in generating two tiles
$g_1$ and $g_2$ which correspond to the same tile $c \in C$. While the connection in the valence
check prevents this in most cases, double generation is still possible if our
sequence of generations has managed to go around a tile $c \in C$ which has not been
generated yet (and in that case, the valence check detects non-uniqueness
after $c$ is generated). Avoiding such inconsistency is the the main issue that 
we need the regular tree structure to avoid.

We solve the inconsistency by using the union-find data structure, that is, we
learn that the first and $v(t(g),i)$ in the cycle are actually the same cell, and
unify them; such an unification may also recursively cause unification of other cells.
Our union-find data structure is a bit more complex than the usual one -- we don't
only learn that $g_1$ and $g_2$ of type $t$ are the same cell, but also
the orientation matters. Thus, the unification method {\sc Unify} is called for a pair of walkers
$(g_1,i_1)$ and $(g_2,i_2)$ (it may happen that $i_1=i_2+rn_t$ for some
$r \neq 0$). Every cell $g \in G$ remembers the representative $u(g) \in G' \times \bbZ$,
which means that the treewalker $(g,0)$ has been unified with $u(g)$.
The unification method {\sc Unify} works as follows:
\begin{itemize}
\item make sure that $g_1$ and $g_2$ are the current representatives, i.e.,
$u(g_1)=(g_1,i_1)$ and $u(g_2)=(g_2,i_2)$ (if not, find their representatives)
\item rotate the second treewalker so that $i_2=0$, and rotate the first treewalker accordingly
\item set $u(g_2)$ to $(g_1,i_1)$
\item move all the connections of $g_2$ to $g_1$, recursively unifying them in the case if we already knew a (different) neighbor 
for both $g_1$ and $g_2$
\end{itemize}

\subsection{Calculating distances from $g_0$}\label{sec:shortcut}

Our algorithm also needs to know $\delta(g)$ for every generated tile, where $\delta(g)$
is the length of the shortest path from $g_0$ to $g$. This is again a major challenge, because
what we need is not the length of the shortest path in $G$, but the length of the
shortest path in $C$ (i.e., taking not generated tiles into account). Again, solving this
challenge is one of the main issues that we generate a geodesic regular tree
structure to avoid. Standard graph algorithms such as Breadth First Search are
impractical because of the exponential expansion of $\bbH^2$.

We simply assign the distance of $\infty$ to every tile we generate, and whenever
we find out that two tiles $g_1$ and $g_2$ are adjacent and the currently known
value of $\delta(g_2)$ is greater than the currently value of $\delta(g_1)+1$,
we set $\delta(g_2)=\delta(g_1)+1$ and recursively check all the known neighbors of
$g_2$. In case of unification, we set $\delta$ for the new cell to 
$\min(\delta(g_1), \delta(g_2))$.

If the rest of the algorithm needs to know $\delta(g)$, we mark $g$ as \emph{solid}.
In case if $\delta$ changes for a solid tile, we have a \emph{distance error}.
We note this fact -- we know that the
results obtained so far are not reliable, and we will need to redo a part of our
computations.

Additionally, we try to prevent the same happening in the future by recording
\emph{shortcuts}. We use the following method. When we set $\delta(g_2)$ to
$\delta(g_1)+1$, we also save the direction from $g_2$ to $g_1$ in memory.
By checking the saved directions recursively, we can thus generate the whole
path of length $\delta(g_2)$ from $g_2$ to $g_0$. In the case of a distance
error, we get two paths: the old one $\pi_1$ and the new, shorter one $\pi_2$.
Let $g_3$ be the first intersection of $\pi_1$ and $\pi_2$, and $p_1$ and
$p_2$ be the sequence on turns of paths $\pi_1$ and $\pi_2$ until $g_3$.
Let $p$ be the loop obtained by concatenating $p_1$ and the reverse of $p_2$.
We record the loop $p$ for the type $t(g)$. 

Later, whenever the rest of the algorithm asks for $\delta(g')$, we loop over
all recorded shortcuts $p$ for $t(g')$, and call the procedure {\sc TryShortcut}$(g',p)$.
This procedure tries to replicate the sequence of treewalker movements given by $p$,
but starting from $g'$. In the case when during this replication
we get to a move which leads to a yet ungenerated connection $e(g'',i)$ = NULL,
we check whether continuing $p$ would yield a shortest path to $g'$, by 
comparing $\delta(g')$ with $\delta(g'')+n$, where $n$ is the number of remaining
steps in the loop $p$. If no, we exit the procedure; if yes, we continue, eventually
reaching $g'''$. In case if $g''' \neq g'$ we know that $g'$ and $g'''$ should be
unified, and thus, we have obtained a shorter path to $g'$ than the one previously known.

When a new shortcut is added, we also call {\sc TryShortcut}$(g',p)$ for all the solid
cells of type $t(g)$.

We compute the parent of every tile according to the rules from Section \ref{sec:parent}.
In most cases, applying the first rule will be sufficient to yield a winner. Since this
rule is local, the winner can be determined very quickly. We cache the parent direction
for every tile $g \in G$; in case if $\delta$ is updated for some $g$, we clear the cache for
$g$ and its neighbors. The same shortcut method is also used when we find out that 
the information about the parent of $g$ we have been using was incorrect.

\begin{theorem}\label{nondivine}
If $G$ and $\delta$ approximate $C$ as above, and the algorithm finishes successfully,
the GRTS obtained for $C$ is correct.
\end{theorem}

\subsection{Other optimizations and special cases}\label{opts}

The running time of our algorithm depends of its implementation details. In this section we list
the most important details for our implementation.

\begin{itemize}
\item Rather than keeping the set $S$, we keep the list of {\it important} tiles $S'$; the set $S$
is constructed in every iteration, starting from $S'$, and adding its closure on the fly.
\item We cache $E(g)$ for every tile $g$ for which it has been computed. For every $q\in Q$ we also cache
the list of tiles $g$ such that $E(g)=q$, so that in the case of inconsistency at $q$, we can easily remove
all states $q$ from the cache.
\item After finding the closed set $S$, in every iteration we examine all the inconsistencies and
extend the decision trees according to all of them. Likewise, we 
call {\sc ExamineBranch} for all tuples $(q,i_1,i_2)$ -- however, if the same incorrect transition rule is
found multiple times, we add just the first witness to the list of important tiles.
\item In the case of a distance error, we clear the {\sc GetSide} cache, and restart the iteration.
\item Unification and distance errors may mislead our algorithm, extending branches 
and decision trees further and further. To deal with this, we
periodically refresh all the sets of important tiles and decision trees back to their initial state.
We do this after every $2^n$ iterations ($n \in \bbN$).
\item Similarly, unification and distance errors may mislead {\sc GetSide} and {\sc ExamineBranch}
procedures. If the number of iterations in these procedures exceeds the variable $s$, we restart the
iteration.
\item The algorithm as described above starts the tree from only one root $g_0$. For practical reasons 
(Section \ref{sec:app}) it is
useful to know the tree structures starting from root of every possible tile type $t \in T$. 
Thus, our initial
$G$ and the initial $L$ consists of root tiles of all types, these root types generate separate tessellations
each with its own $\delta$. The obtained tree structure will contain many roots (states without parents)
and descendant states will likely be shared between the various roots.
\end{itemize}

While we have no proof that our methods of dealing with unification and distance issues always
yield a solution (in time polynomial in $|T|$ for bounded degrees of tiles and vertices), our extensive
testing lets us conjecture that it is the case. 

\section{Applications}\label{sec:app}

In this section we briefly list the applications of geodesic tree structures.

{\bf Computing distances from the origin, or an arbitrary fixed tile.}
In a geodesic tree structure, we immediately know the distance from every tile $g$ to the
origin tile $g_0$. We can also compute the origin from multiple fixed tiles $g_0'$ -- simply by
creating another tree structure $T'$ rooted in $g'_0 \neq g_0$, and remembering the mapping between $T'$ and
$T$.

{\bf Computing relative distances from an ideal point.}
In section \ref{sec:trees} we have said that the situation $x=P$ may not happen,
because we start $g_0$ in a state which has no parent.
However, there is also an alternative structure where we start in $g_0 \in G$ which does
have a parent. In this case, we generate an unrooted tree
that is descending infinitely. If we need to extend the tree at $(g,i)$ and find
$x = P$, we connect $(g,i)$ to
a new tile $(g_1,i')$ of state $q(g_1)$
such that $e(q(g_1),i') = q(g)$. We pick $g_1$ randomly from all the possible states
that can be indefinitely extended downwards, and having live branches on both left and right.
Effectively, the tree structure is rooted at 
The set of tiles is the discrete analog of a horocycle.
This generalizes the algorithm used in HyperRogue \cite{hyperrogue} to generate discrete horocycles.

{\bf Computing shortest paths between two tiles.}
It is also important to determine the shortest paths between two arbitrary tiles $g_1$ and $g_2$.

In a Gromov hyperbolic tiling, we know that every shortest path from $g_1$ to $g_2$ lies inside
the $\delta$-neighborhood of the shortest path from $g_1$ to $g_2$ going through $g_0$, where
$\delta$ is a constant depending on the tessellation. This gives us an algorithm to determine
$\delta(g_1,g_2)$ in time $O(\delta(g_1,g_2))$ \cite{trigrid} -- simply generate the shortest path from
$g_1$ to $g_0$ and the shortest path from $g_2$ to $g_0$, and generate $\delta$-neighborhoods of
every tile on the way. For a fixed tessellation, this gives an algorithm with time complexity
$O(\delta(g_1, g_2))$. Finding $\delta$ efficiently is the subject of further work.

\begin{theorem}\label{distperiodic}
In a fixed periodic tiling of Euclidean space $\bbE^n$, we can compute the distance $d$ between two tiles
(given their coordinates) in time polylogarithmic in $O(\log d)$.
\end{theorem}

The proof based on the computational properties of Parikh images \cite{parikhlics} can be found in the Appendix.

{\bf Representing arbitrary points in $\bbH^d$.} Geodesic tree structures can be used to
represent points in $\bbH^d$ while avoiding precision errors: each point $x_1$ is represented by a tile $g_1$ it is in,
and its coordinates relative to the tile $g_1$ \cite{rtviz}. To transform $g_2$-relative coordinates to $g_1$-relative
coordinates, we need to determine the shortest path from $g_2$ and $g_1$, and compose the isometries along this path.
This is straightforward in a periodic tessellation.

{\bf Coordination sequences.}
The coordination sequence of a tessellation with the starting tile $g$ is the sequence $a_0, \ldots$ such
that $a_n$ is the number of tiles in distance $n$ from $g$. Coordination sequences are studied in the foundations of crystallography,
see \cite{colorbook} for an extensive literature. A regular tree structure rooted at the tile $g$ lets us easily 
find the system of linear recursive formulae for the coordination sequence: for a tree state $q\in Q$, let $a_q(n)$ be the
number of tiles in distance $n$ in state $q$. For $n>0$ we have $a_q(n) = \sum_{p \in Q} a_{p}(n-1) d(p,q)$, where
$d(p,q)$ is the number of children of $p$ in state $q$.

\section{Experimental results}\label{sec:experimental}
We have tested our algorithm on 149629 tessellations, including all the Euclidean (102251) and hyperbolic (47378) tessellations 
from \cite{tescatalog}%
\footnote{Our software, tessellation data and experimental results can be found at \url{https://figshare.com/articles/software/Generating_Tree_Structures_for_Hyperbolic_Tessellations_code_and_data_v2/19165922}.}.
We analyzed average running time (in seconds on Intel(R) Core(TM) i7-9700K CPU @ 3.60GHz), memory consumed by the algorithm
(counted as the number of tiles generated), and the number of tree states. Figure~\ref{correct_time_geo} depicts the results in division by geometry. In Appendix, Figure~\ref{correct_shps} depicts division by the number of shapes. 
The are no hyperbolic tessellations that took longer than 10 seconds, and in the case of Euclidean geometry, only 44 tessellations took longer than 30 seconds.

\begin{figure}[ht]
\begin{center}
  \includegraphics[width=0.3\linewidth]{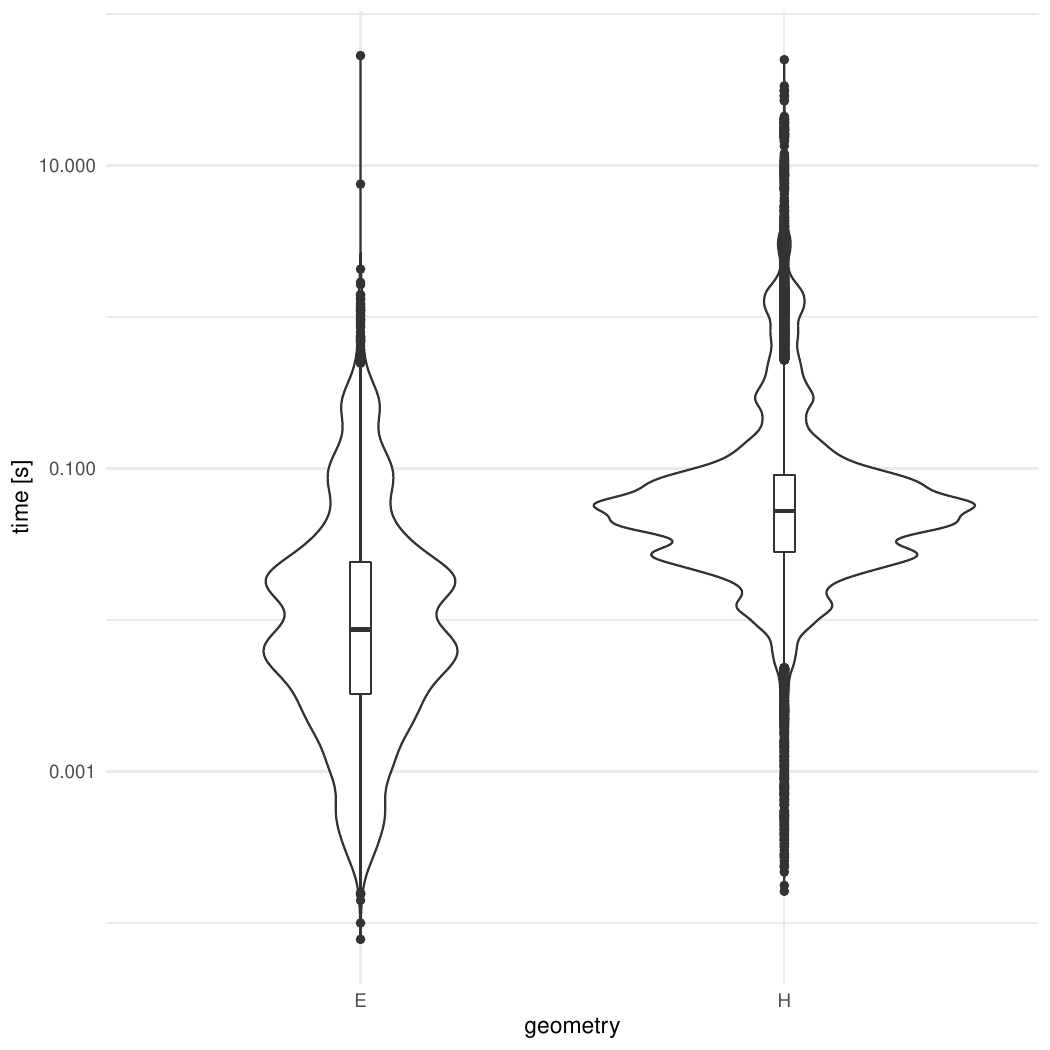}
  \includegraphics[width=0.3\linewidth]{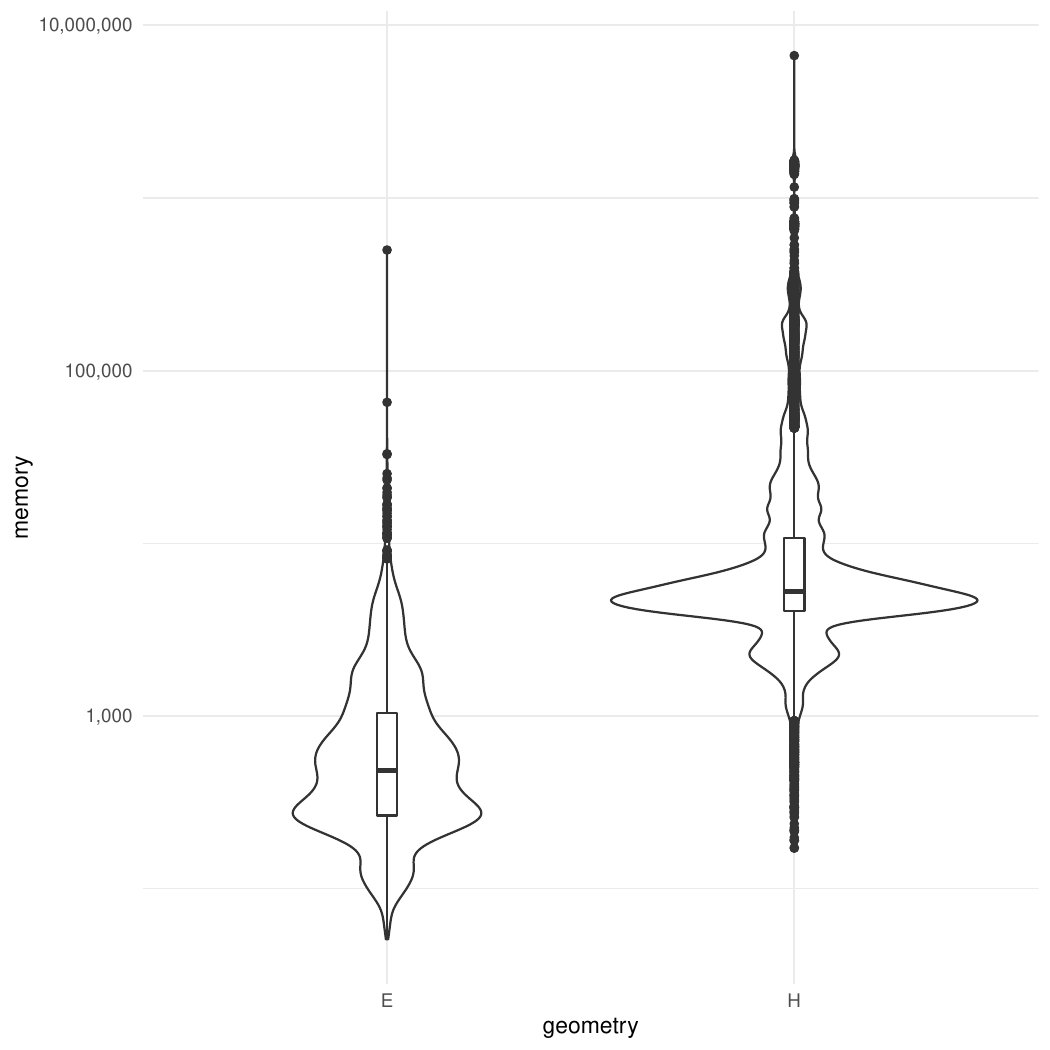}
  \includegraphics[width=0.3\linewidth]{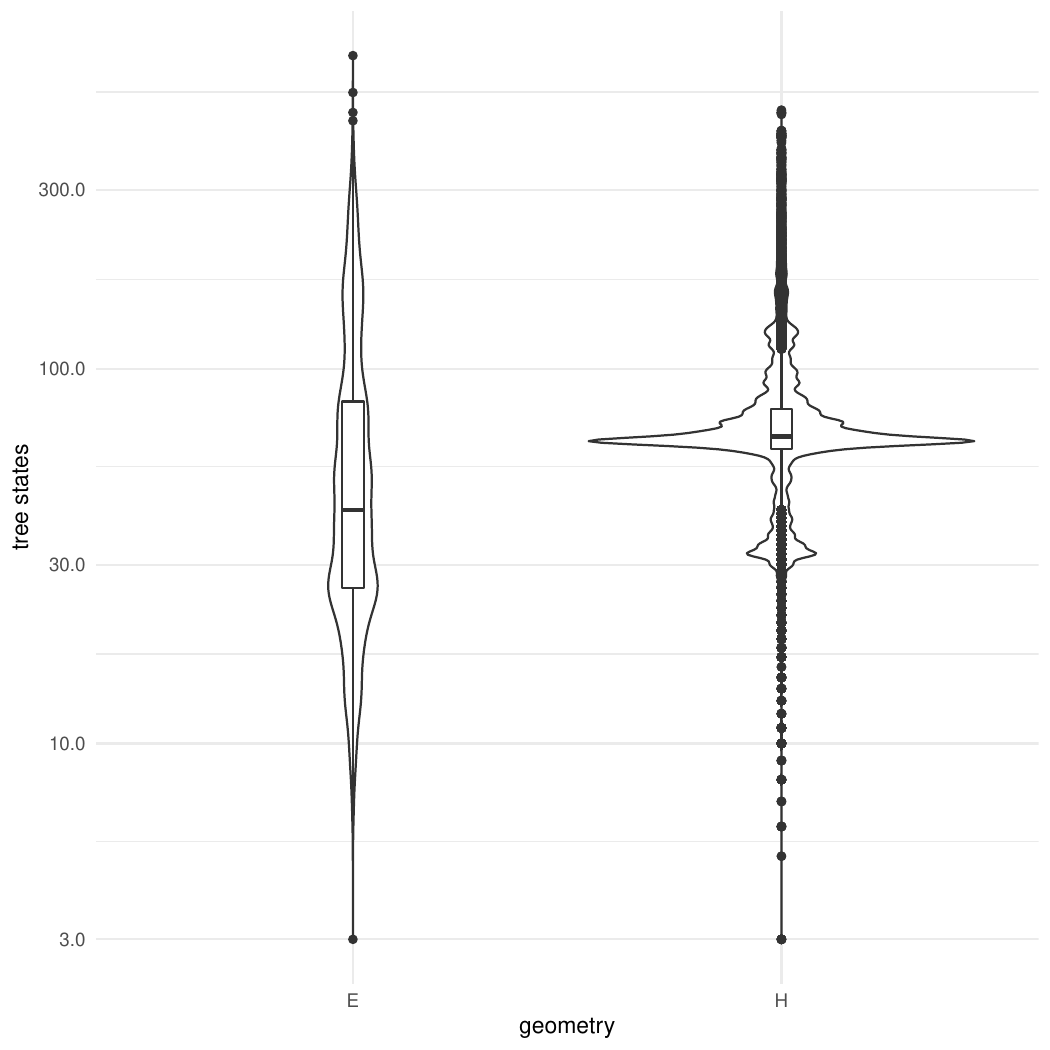}
\end{center}
\caption{ Average running time [s], memory consumed,
and the number of tree states in division by geometry.
\label{correct_time_geo}}
\end{figure}

No matter the characteristic analyzed, Euclidean tessellations need more resources on average, however this result significantly stems from the difficulty of the tessellations we analyzed. Interistingly, the distributions for hyperbolic tessellations are steeper. Whereas one expects that the number of shapes correlates positively with
the amount of resources, we notice that the usage of the higher number of shapes
comes with decreasing variance for hyperbolic geometry (Fig.~\ref{correct_shps} in Appendix: note the triangular pattern -- the range for tessellations based on low number of shapes is much greater than the range for tessellations with the highest number of shapes). When it comes to Euclidean tessellations, we notice that no matter the characteristic, the range of values for small number of shapes (1-8) is significantly lower than for large number of shapes. Moreover, within those two groups, the ranges are similar -- we do not notice the triangular shape known from the results for hyperbolic geometries.
Appendix contains Figures \ref{worstcase:h} and \ref{worstcase:e} that show the tessellations which achieve the highest time, memory, and number of tree states.

Our dataset includes all $k$-uniform tilings of the Euclidean plane for $k\leq 12$. These are tilings by regular polygons with $k$ orbits of vertices (under isometries including orientation-changing isometries).
There is a specific family which achieves the largest number of states for every $3\leq k\leq 12$, consisting of a row of squares followed by $\Theta(k)$ rows of triangles. This number of states grows
quadratically with $k$. Changing the square into a hexagon (by changing its $s_t$ from 2 to 3) yields a family of similarly difficult hyperbolic tessellations (one state less). Thus, we conjecture that
the number of states is quadratic with the number of tile types when tile and vertex degrees are bounded by a fixed constant.

We have also compared the performance of our implementation to other approaches. Below we present a short summary of our results; see Appendix for more details.
BFS, numerical unification, and Knuth-Bendix procedure all achieve significantly worse results.
Running our algorithm on $G=C$ and correct values of $\delta$ yields better results, but not much so.
We conclude that our algorithm is successful at dealing with unification and distance errors.
The shortcut recording optimization described in Section \ref{sec:shortcut} is crucial to this success.




\section*{Acknowledgments}
This work has been supported by the National Science Centre, Poland, grant UMO-2019//35/B/ST6/04456.

\newpage

\appendix

\section{Omitted proofs}
\begin{proof}[Proof of Theorem \ref{distperiodic}]
In a fixed periodic tiling of the Euclidean plane, the group of isometries of the tessellation $G$ is one of 
the 17 wallpaper groups, and has a subgroup $G_1$ of translations of finite index. 
By considering two tiles $g$ and $g'$ to be of the same type only if there exists an isometry in $G_1$
which takes $g$ to $g'$, and thus, we can assume without loss of generality that,
for every tile type $t \in T$, the set of all tiles of type $t$ forms a lattice.
For two tile types, $t_1$, $t_2$, The set of $(x,y,d) \in \bbZ^3$ such that the number of steps between the copy of $t_1$ 
at lattice coordinates $(0,0)$ and the copy of $t_2$ at lattice coordinates $(x,y)$ is at least $d$ 
is a Parikh image of a finite automaton, and thus by Parikh theorem \cite{parikh} a semilinear set. 
 Membership in such a semilinear set can be recognized in time polylogarithmic in $x+y+d$
\cite{parikhlics}, and by easy reduction we can also find the smallest $d$ for given $(x,y)$ in polylogarithmic time.
The same argument holds in higher dimensions (the existence of $G_1$ follows from the Bieberbach theorem \cite{Bieberbach1911,Bieberbach1912}).
\end{proof}

\section{Further experimental results}
\subsection{Na\"ive approaches}
We will compare alternative approaches on a reduced dataset, excluding $k$-uniform Euclidean tilings for $k \geq 10$. 
As a result, the final sample contains 16132 Euclidean tessellations and 47378 hyperbolic ones. The distributions of shapes per tessellation resemble distributions in the full sample (Fig.~\ref{shapes_geo}).
We allow the alternative approaches to generate 80,000,000 tiles, run for 600 seconds,
and run 9999 iterations of the main loop. These limits are significantly greater than what our main algorithm achieves (outliers are 2700850 tiles, 32 seconds, 2346 iterations).


{\bf Single Origin.}
As explained in Subsection \ref{opts} we start from roots of every possible type $t \in T$. This increases
our test coverage and has practical advantages.
We can also start from only single root (chosen in an arbitrary way). The na\"ive approaches explained above are also based on single origin,
for the sake of simplicity and memory usage reduction.

Since the measurements of running time tend to be unstable, we compare the number of dominating operations. We take the
{\it move}, i.e., the operation of finding $i$-th neighbor of a tile $t$, as the dominating operation.
Usage of single origin rarely worsens the performance of the algorithm in terms of move count for Euclidean tessellations (0.18\% of cases); for hyperbolic tessellations, the share of worse results due to usage of single origin is noticable: 35.52\% of cases (Figure \ref{single_origin}).
Generally, it should need no more memory than the original proposition (63.91\% of the hyperbolic tessellations and 100\% of the Euclidean tessellations). Single origin also comes with lower number of tree states
(there were 8 cases where there were less states in the multiple origin version, probably due to unification/distance errors).

\begin{figure}[ht]
\begin{center}
  \includegraphics[width=0.3\linewidth]{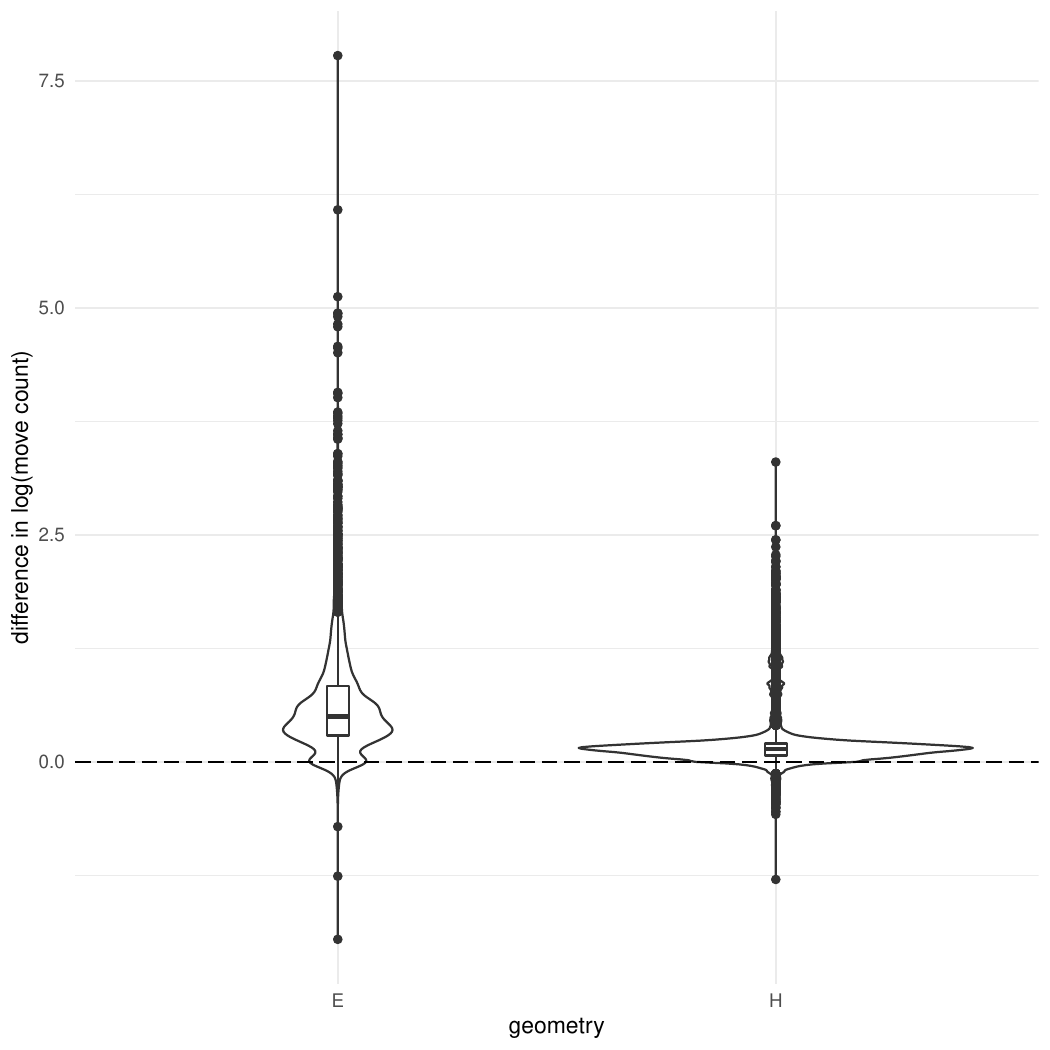}
  \includegraphics[width=0.3\linewidth]{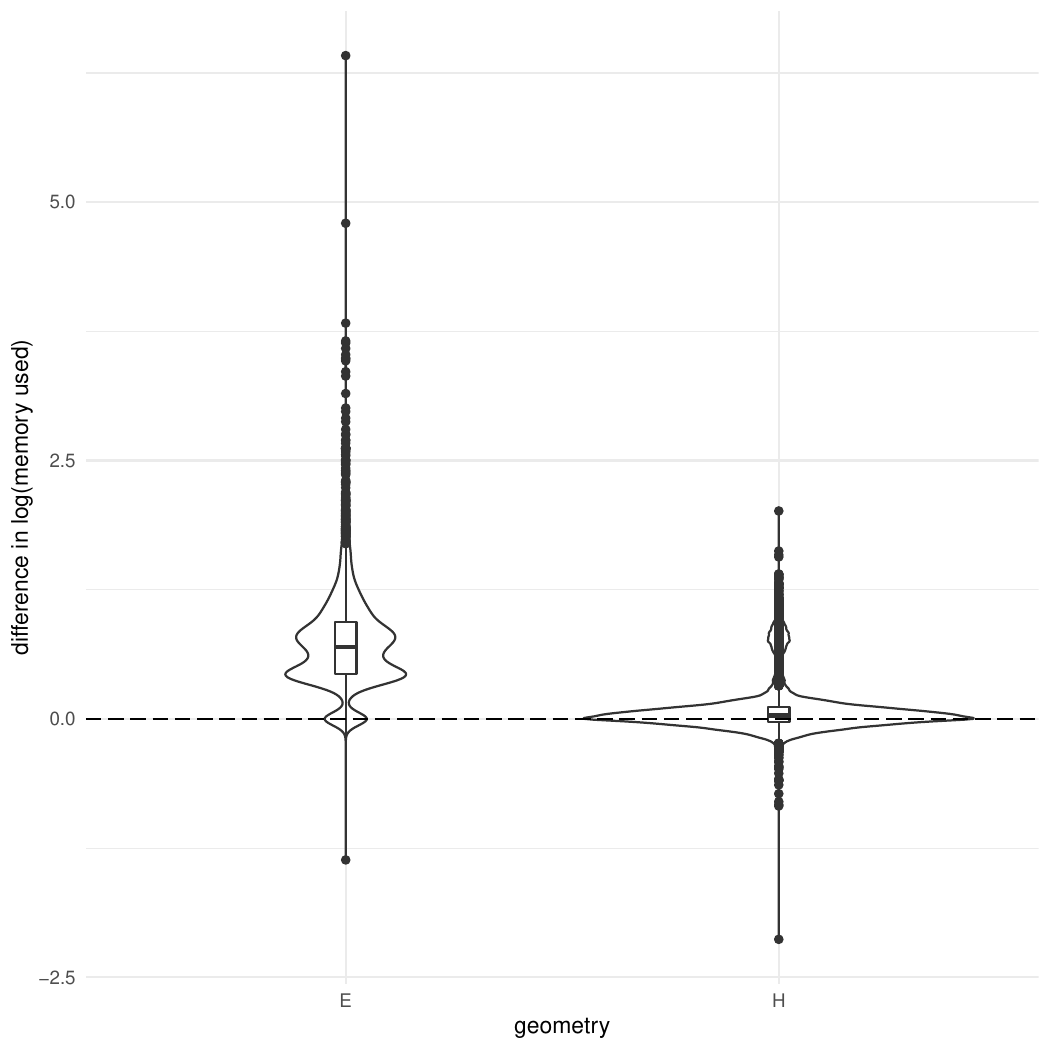}
  \includegraphics[width=0.3\linewidth]{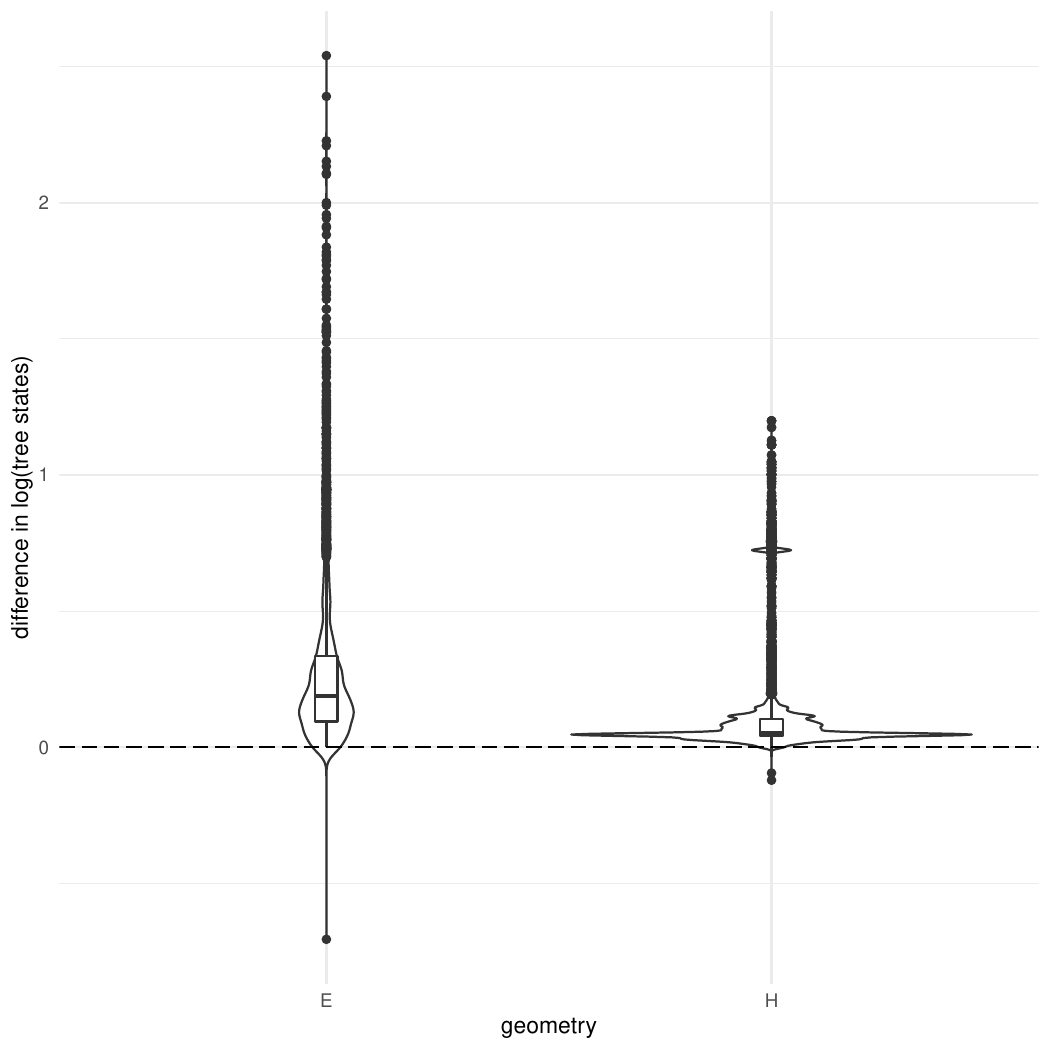}
\end{center}
\caption{ Difference in the number of moves, memory consumed,
and the number of tree states in division by geometry if single origin is used or not.
\label{single_origin}}
\end{figure}


{\bf Numerical implementation.}
One serious difficulty in our implementation is the unification error. A na\"ive attempt of avoiding this kind of error
is to compute the coordinates of every tile. If we get the same coordinates
as a previously existing tile, we know that we connect to it. See Section \ref{sec:previous} for an overview of this
approach. However, we use a modified version which prevents dealing with large values.
We use a tessellation for which the tree structure is known (the \{7,3\} tessellation from Figure \ref{examplefig}a), and we
compute the isometry coordinates relative to it \cite{rtviz}.

We compare two matrices $T$ and $U$ by checking the hyperbolic distance from $Th_0$ to $Uh_0$, and from $Tv$ to $Uv$,
where $h_0$ is the origin, and every vertex $v$ of the tessellation. If the distance is over $10^{-2}$, we assume the
points are different; if it is over $10^{-3}$, we stop computation and report a numerical precision issue. We encountered 586 cases of precision errors,
all of them occured in hyperbolic tessellations.

We have noticed that one of the reasons why our algorithm works better than the numerical approach is the
valence checks (see Section \ref{sec:nogrts}). The na\"ive numerical approach does not perform the valence checks,
which causes it to not be aware of some connections. Even controlling for this issue,
the algorithm still yields numerical errors, including 252 cases of precision
errors, again all of them for hyperbolic tessellations, either using low number of shapes (1 or 2) or a large number of shapes (over 31, 3.17\% of cases).

{\bf BFS.}
Another serious issue is distance errors. A na\"ive attempt of avoiding this kind of error is to use Breadth-First Search algorithm.

This is a very bad idea due to the exponential growth of hyperbolic tessellations.
In our case, we succeded in only 41.55\% of the cases, in 58.45\% of the cases we 
the algorithm stopped because of the overflow of our memory limit.
Note that even if in our simulation we did not encouter distance errors, they are still possible when using BFS because of the unification issues,
although such a situation did not occur in our experiments.

One could also prevent both kinds of errors by using both BFS and numerical implementation. This shares the disadvantages of both approaches.

{\bf Knuth-Bendix method.}
We have run the Knuth-Bendix completion algorithm as outlined in Section \ref{sec:previous} on our data.
Our implementation seems not to terminate for most tessellations. We run the algorithm until
10 seconds pass (14144 cases) or a rewriting rule that has at least 250 symbols is generated (44176 cases); the algorithm
was successful in only 5188 cases. Out of unsuccessful attempts, 75.51\% were hyperbolic tessellations.

\begin{figure}[ht]
\begin{center}
  \includegraphics[width=.45\linewidth]{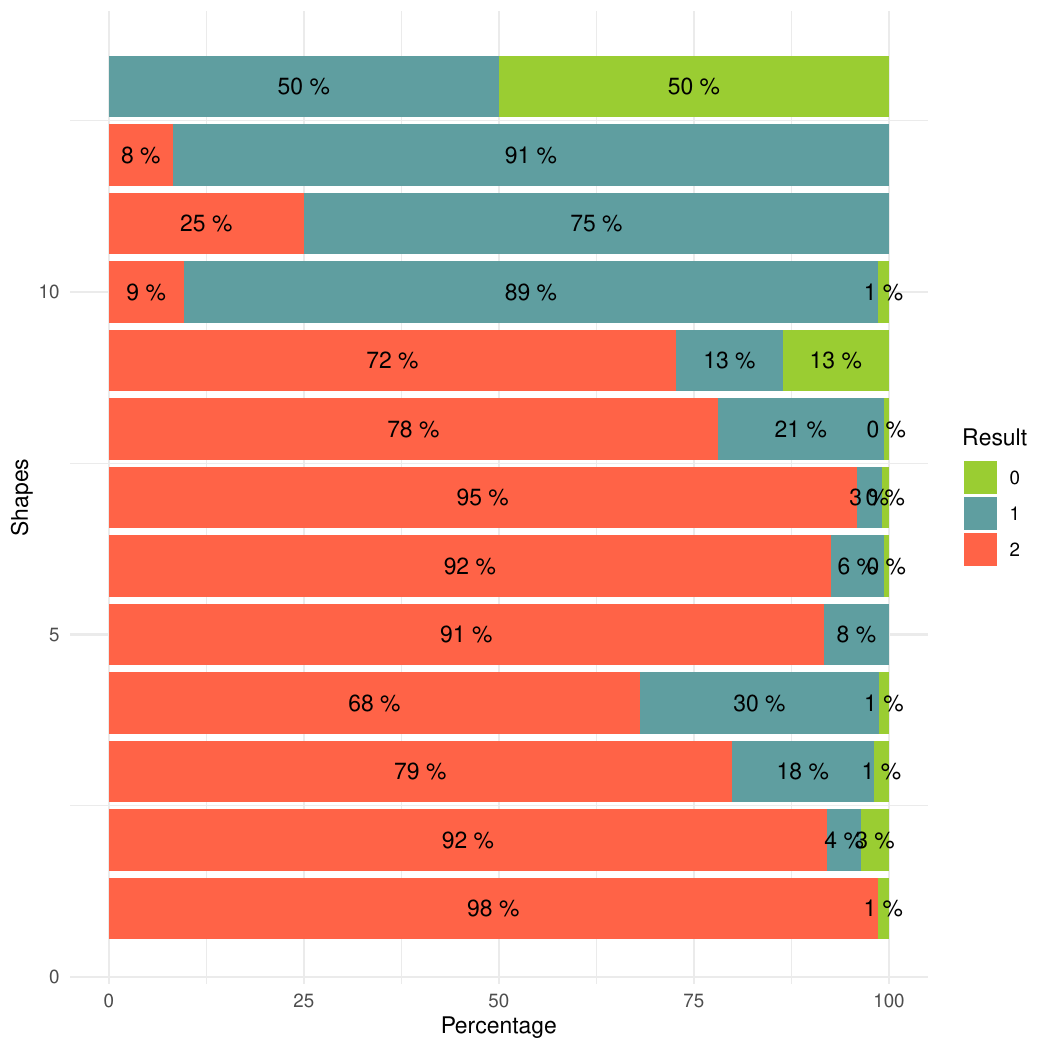}
   \includegraphics[width=.45\linewidth]{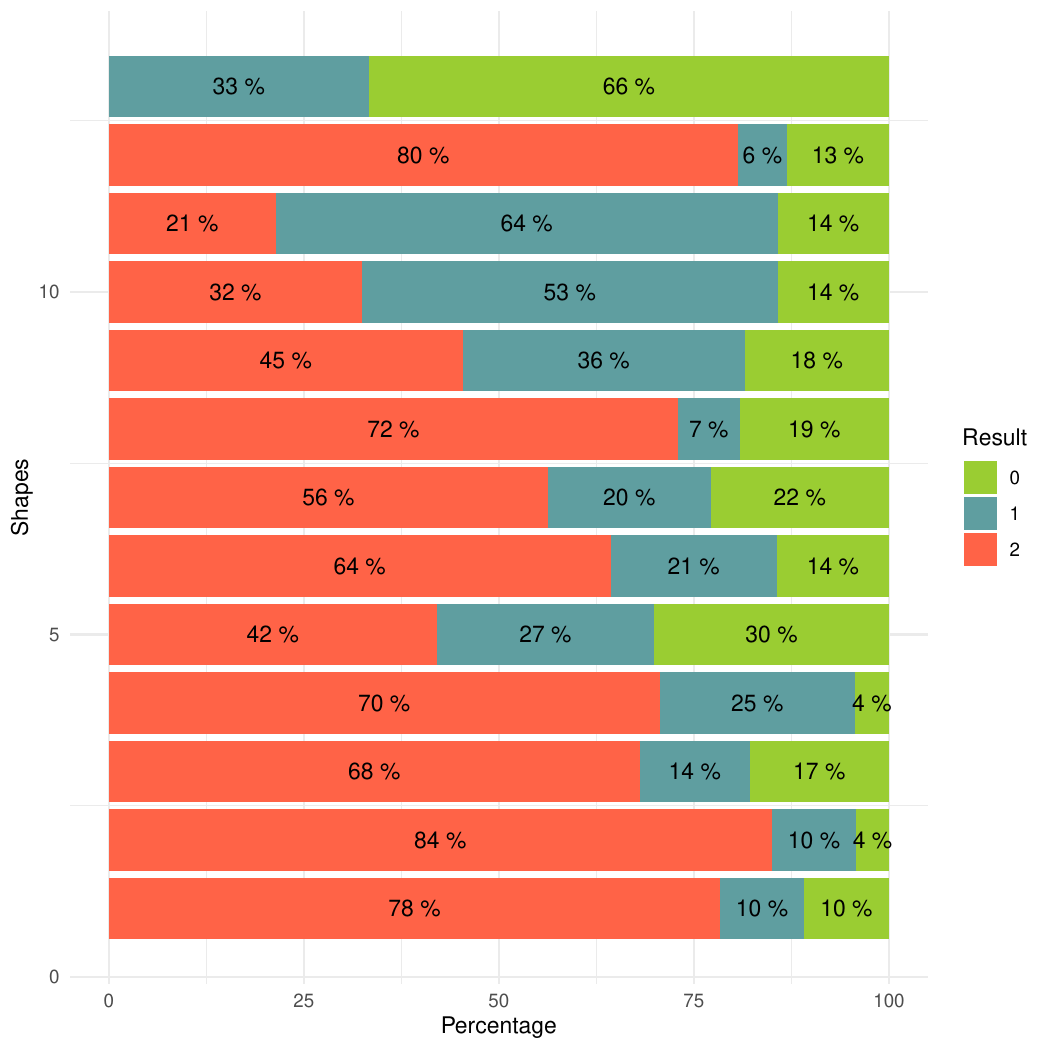}
\end{center}
\caption{Codes in division by geometry and the number of shapes.
\label{kbres_shps}}
\end{figure}

Figure \ref{kbres_shps} shows the percentage of successful cases depending on geometry and the number of shapes
(0 = success, 1 = timeout, 2 = too long rules). The Knuth-Bendix method tends to work better for hyperbolic tessellations (10.74\% of correct cases in comparison to 0.63\% of correct cases). 

When it comes to successful attempts, the algorithm took more than 5 seconds
in only 105 cases and generated temporary rules of length over 150 in only 295 cases, which suggests that the algorithm
does never terminate in most failed cases in our database.

\subsection{The effect of unification and distance errors}

To measure the effect of unification and distance errors, we run the algorithm twice: in the
first run we determine the correct tree structure, and in the second run, we use the correct tree obtained to generate the 
tessellation correctly, and see the performance of our algorithm given the correct data. There are two approaches: (divine) use the
correct tiles and distances (avoiding both unification and distance errors), (demigod) use the correct tiles but do not use the information
about distances (avoiding unification errors, but not distance errors). In both cases we perform the valence checks.

\begin{figure}[t!]
\begin{center}
  \includegraphics[width=0.35\linewidth]{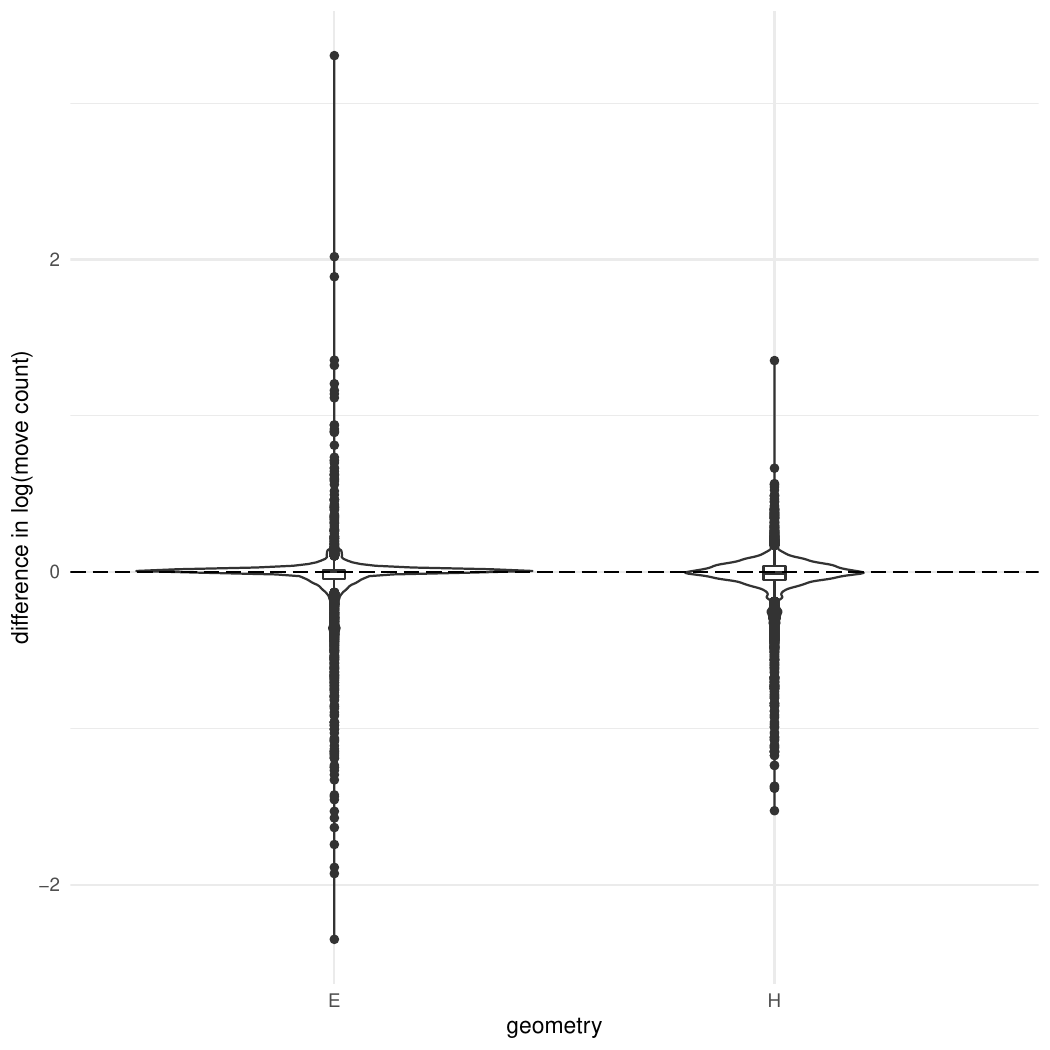}
  \includegraphics[width=0.35\linewidth]{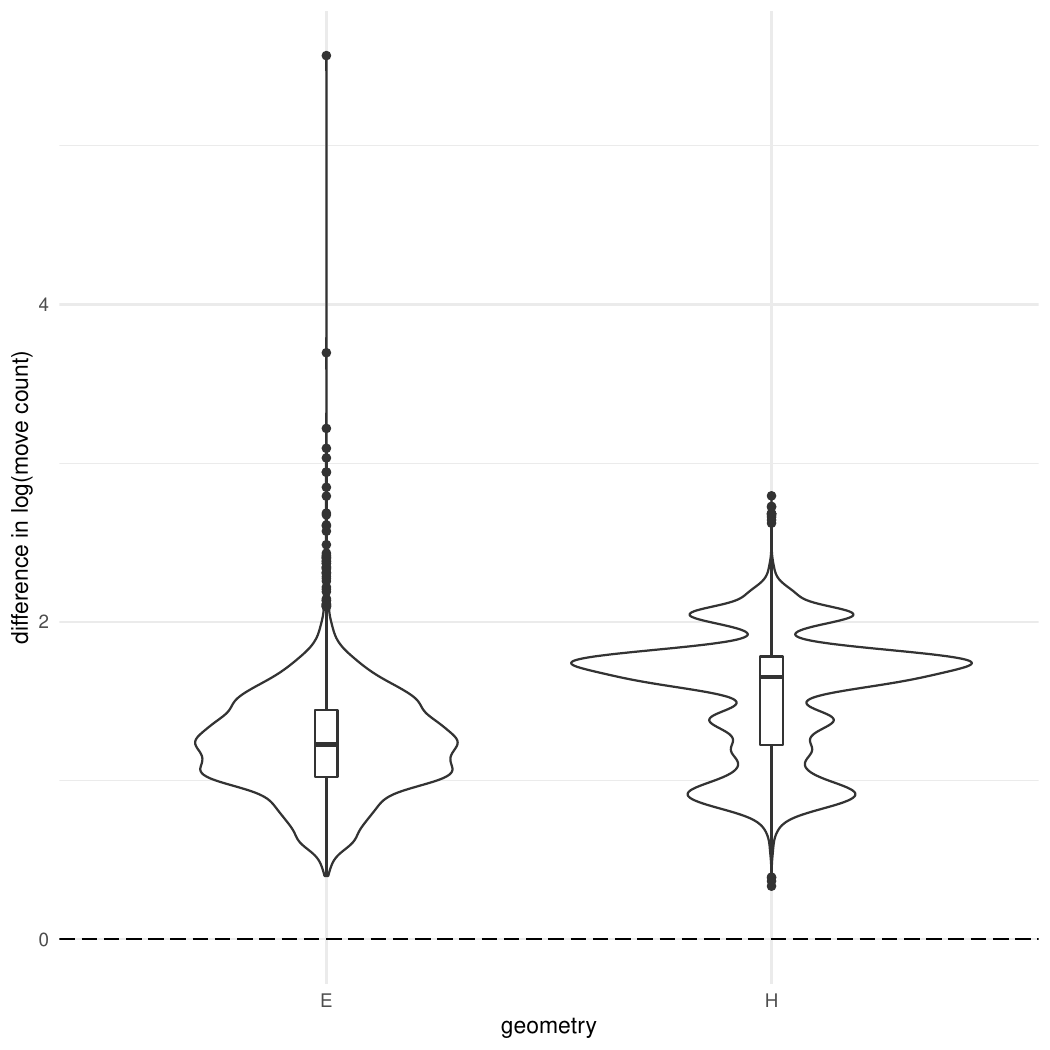}
\end{center}
\caption{ Difference in the log(move count) if demigod (left) or divine (right) are used. Values greater than 0 mean that 
demigod or divine was better than our approximation.
\label{demigod_divine}}
\end{figure}

Figure \ref{demigod_divine} shows the effect of unification and distance errors by comparing the move count of our algorithm 
to the move count achieved by ``demigod'' and ``divine'' versions. In the case of demigod, the mode is close to zero, suggesting that  unification errors occur less frequent than we supposed. The divine approach saves moves in our implementation. We conclude that our algorithm is successful at dealing with unification and distance errors.

The shortcut recording optimization described in Section \ref{sec:shortcut} is crucial to this success.
If we do not generate shortcuts, serious distance errors make our algorithm run noticeably worse. The algorithm fails in 17.8\% of 
Euclidean tessellations and 1.8\% of hyperbolic tessellations.
Resets mentioned in Section \ref{opts} are also important. Without resets, the algorithm fails in 22 cases. That happens only in Euclidean tessellations.

\section{Figures}

This appendix contains the Figures which could not fit in the page limit.
Figures \ref{worstcase:h} and \ref{worstcase:e} show the tessellations which achieve the highest time, memory, and number of tree states.
In each geometry, the selection includes the tessellation which
achieves the worst case across all tessellations. In many cases, the next top places are achieved by tessellations of the same family, which tend to be very similar. Therefore, the remaining
tessellations in Figures \ref{worstcase:h} and \ref{worstcase:e} are those which also achieve high values while being very different.

\begin{figure}[ht]
\begin{center}
  \includegraphics[width=0.3\linewidth]{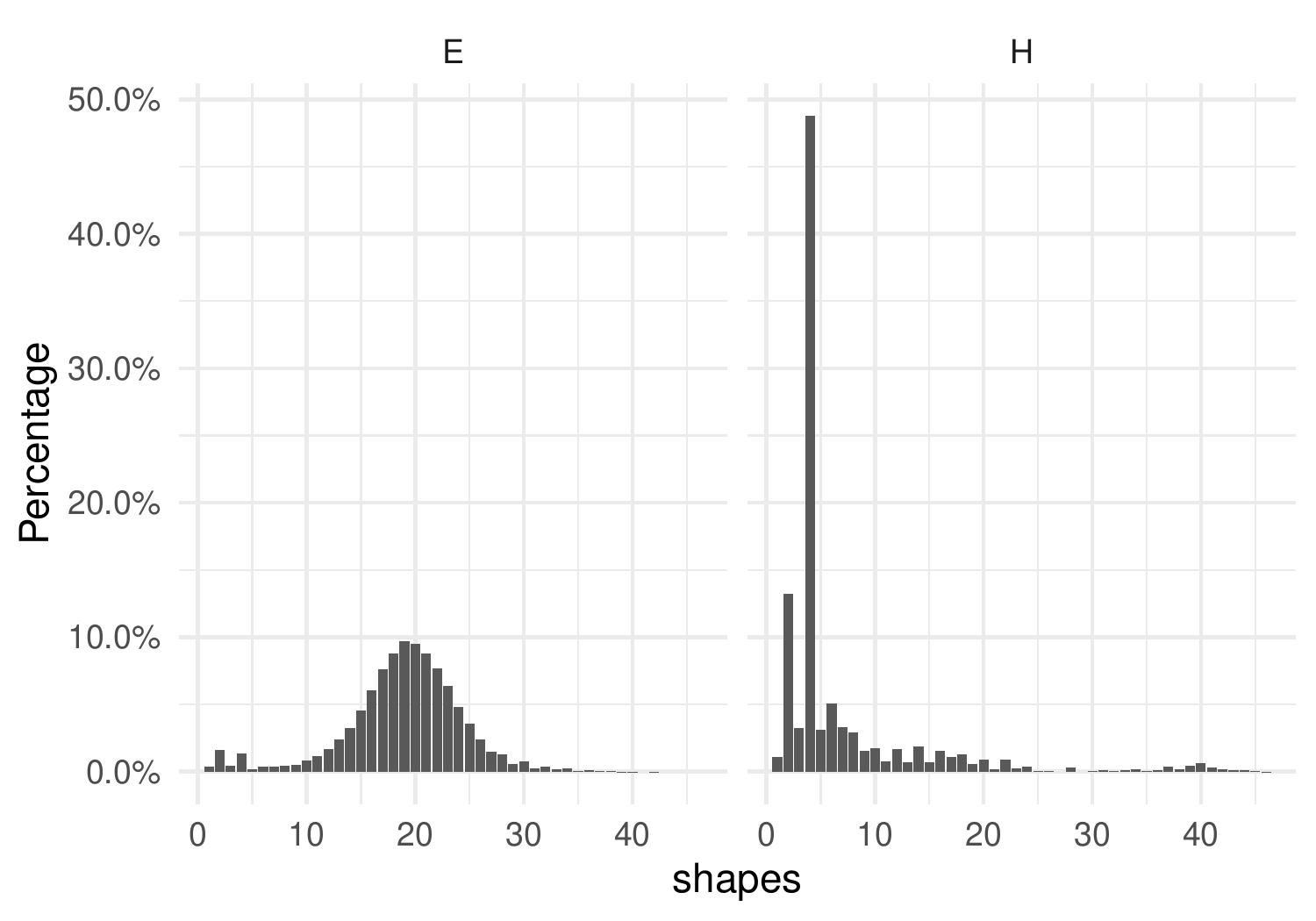}
  \includegraphics[width=0.3\linewidth]{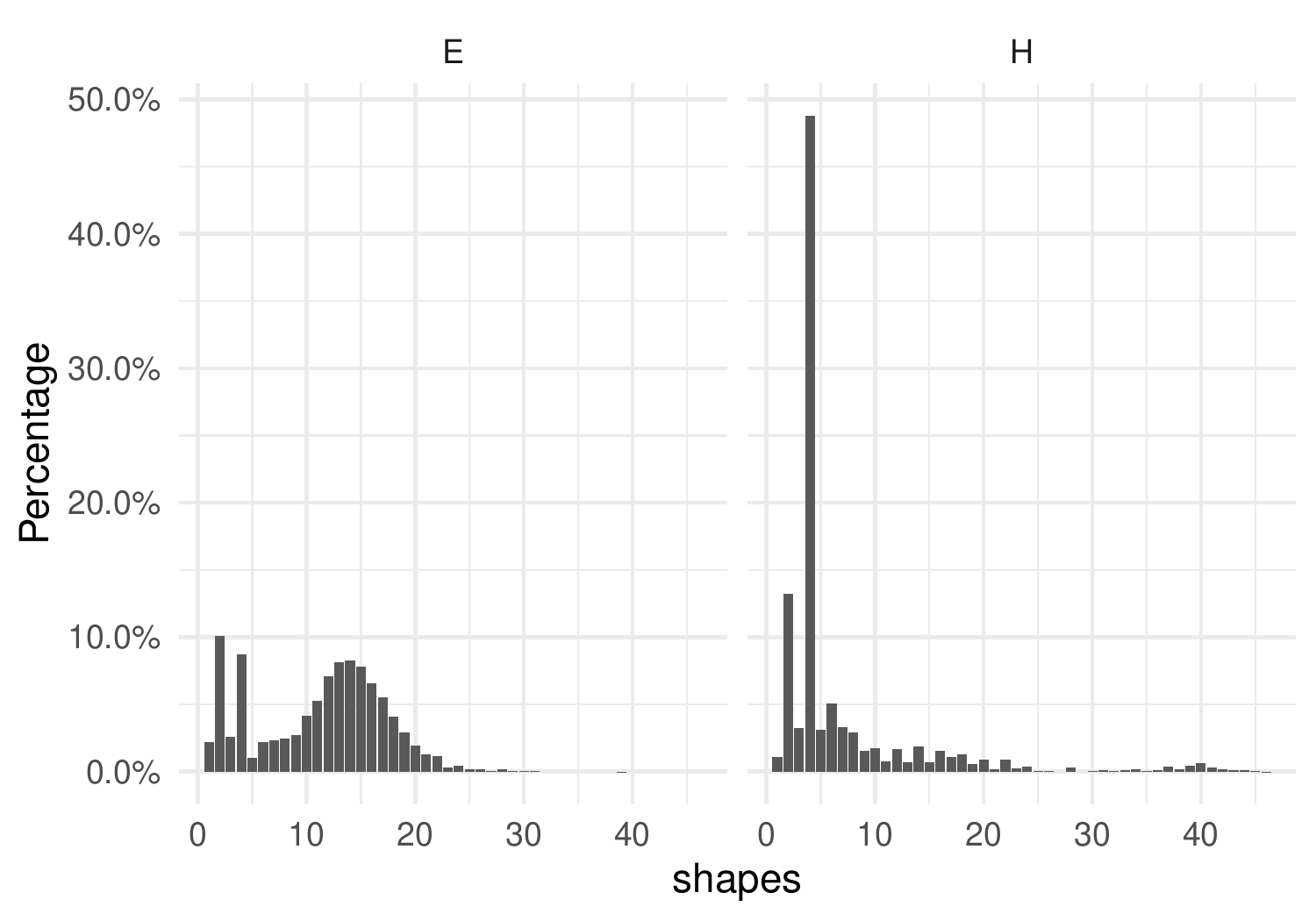}
\end{center}
\caption{Distributions of number of shapes per tessellation for full data (left) and the experimental subset (right) in division by geometry.
\label{shapes_geo}}
\end{figure}

\begin{figure}[h!]
\begin{subfigure}{\textwidth}
\begin{center}
  \includegraphics[width=.37\linewidth]{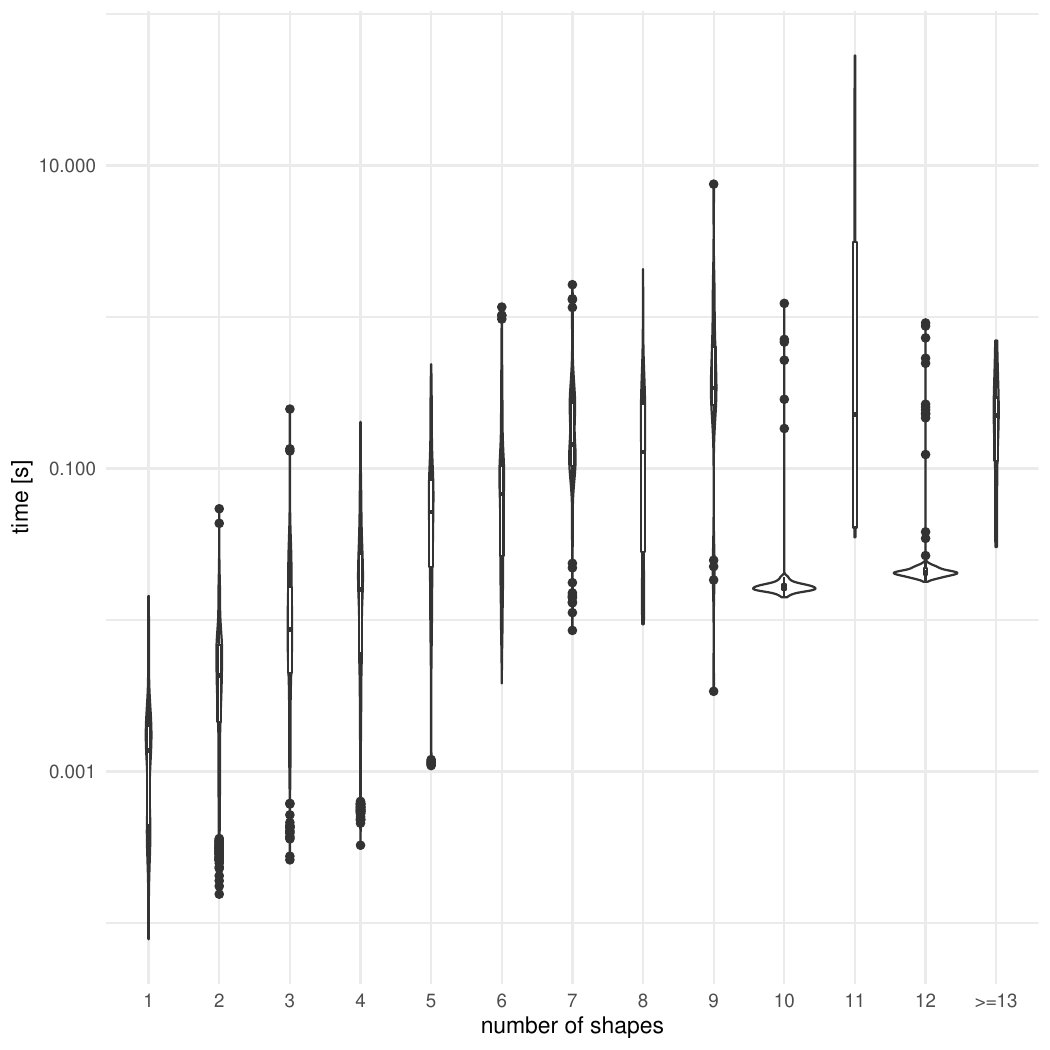}
   \includegraphics[width=.37\linewidth]{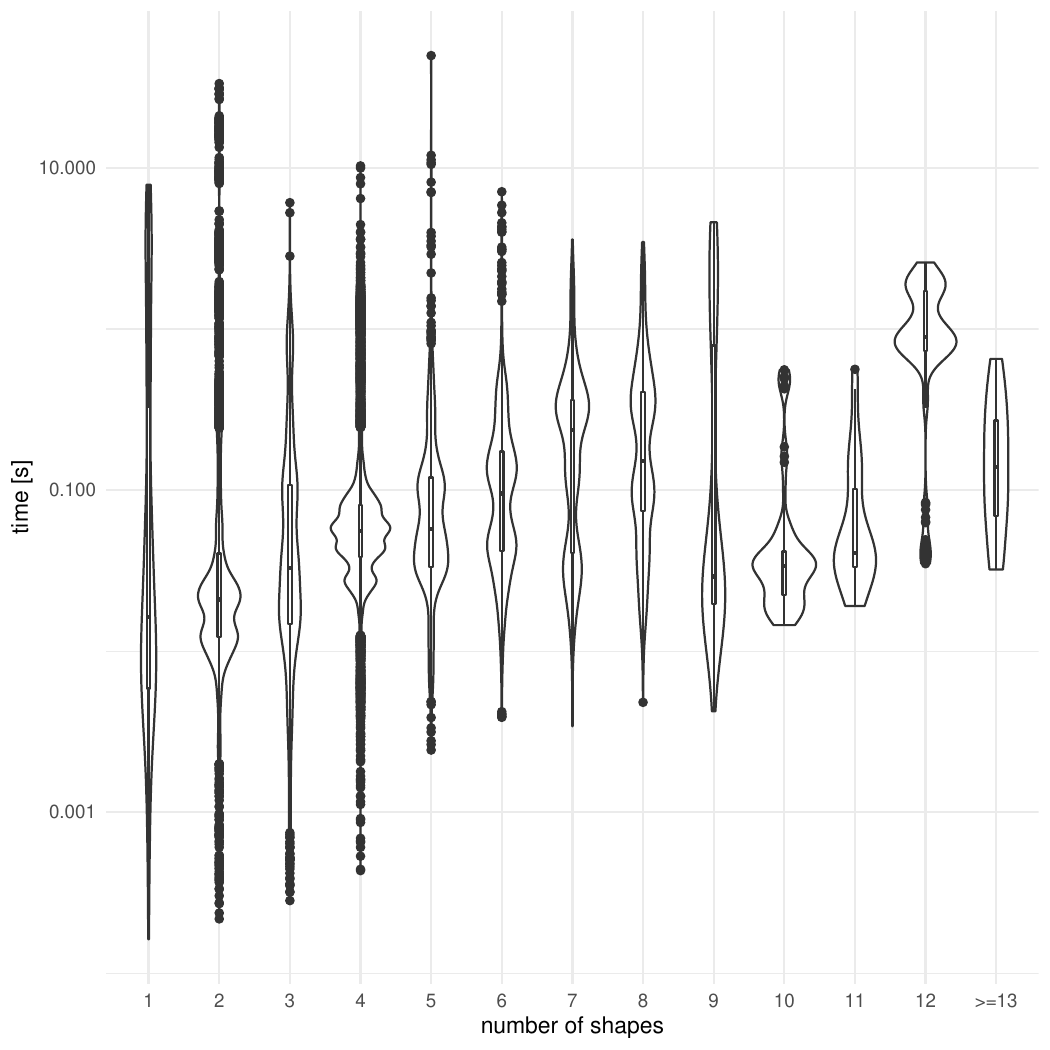}
\end{center}
\caption{Average running time [s]}
\end{subfigure}
\begin{subfigure}{\textwidth}
\begin{center}
  \includegraphics[width=.37\linewidth]{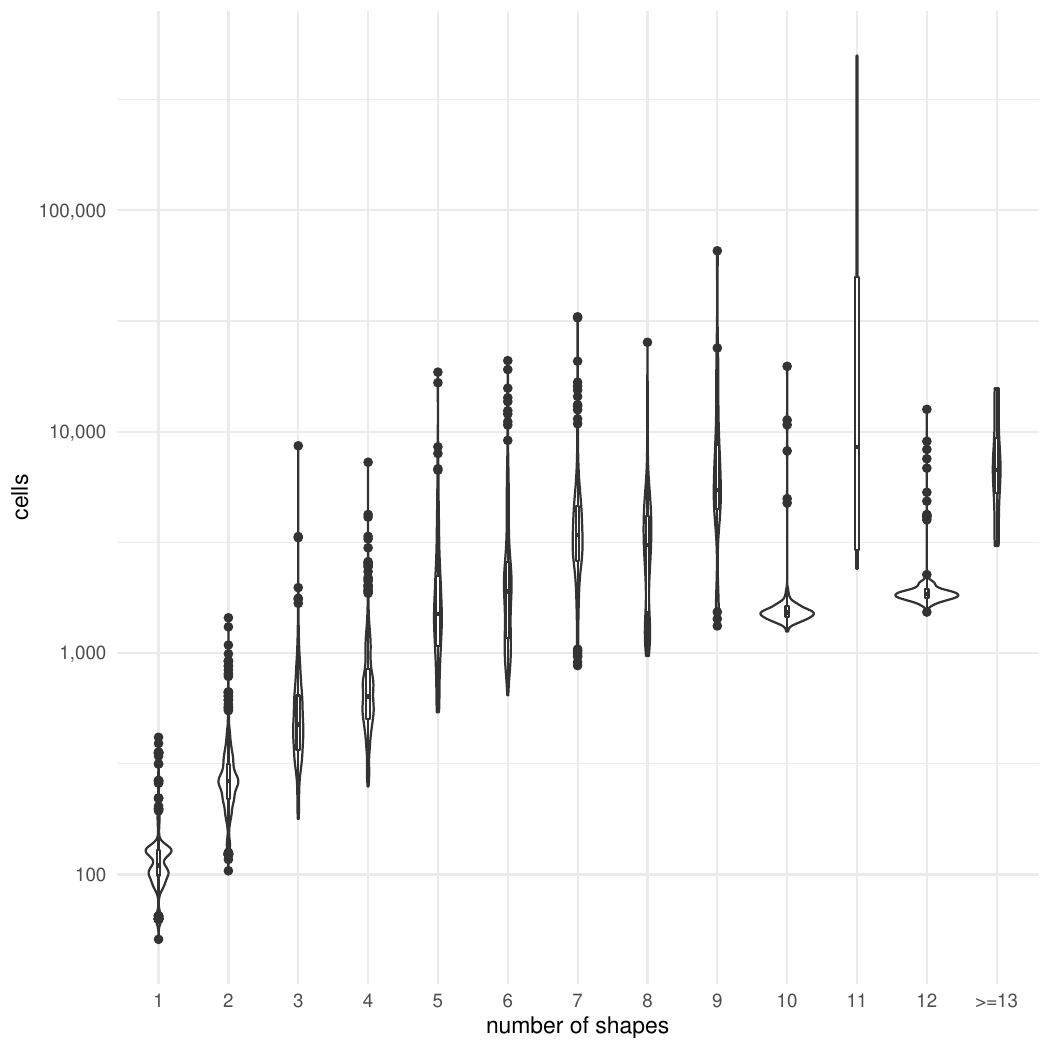}
   \includegraphics[width=.37\linewidth]{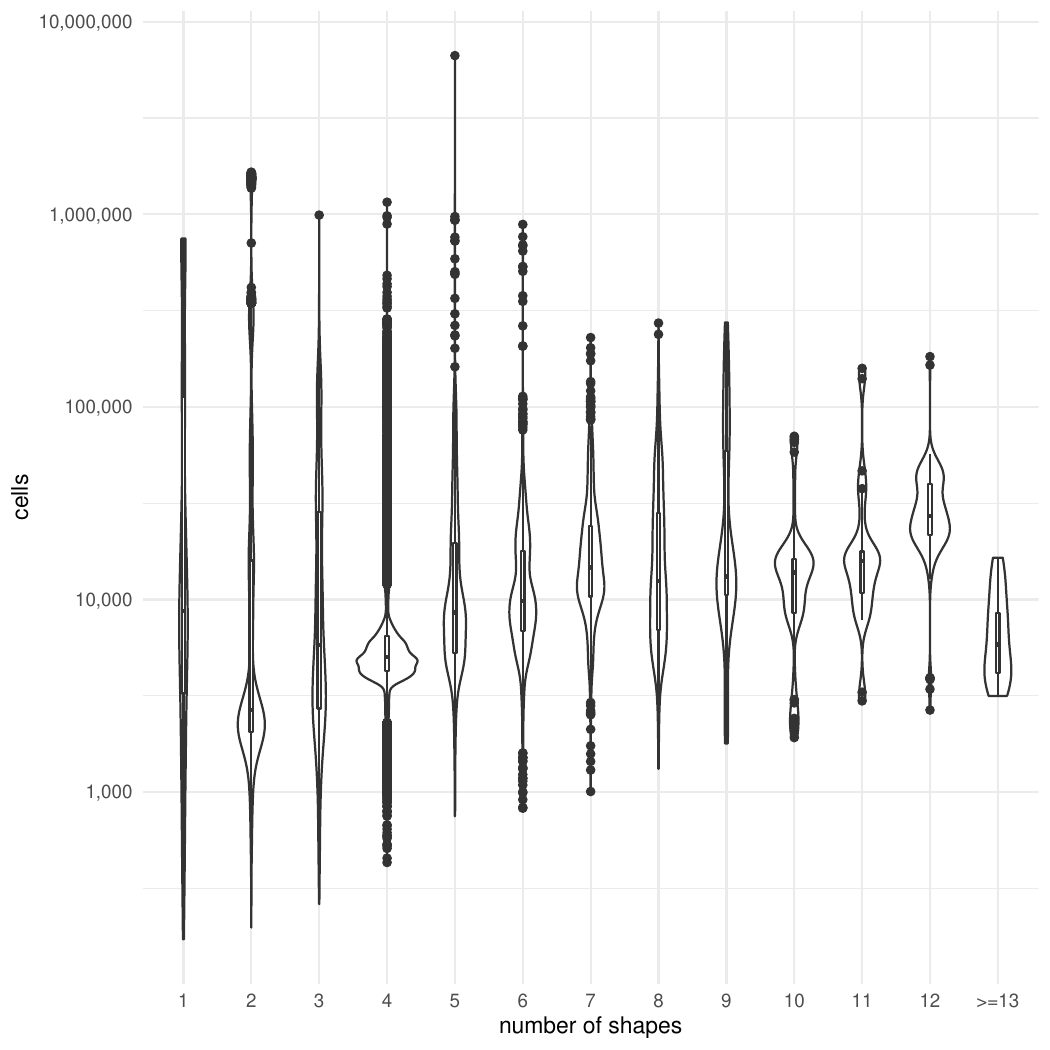}
\end{center}
\caption{Memory use (number of cells)}
\end{subfigure}
\begin{subfigure}{\textwidth}
\begin{center}
  \includegraphics[width=.37\linewidth]{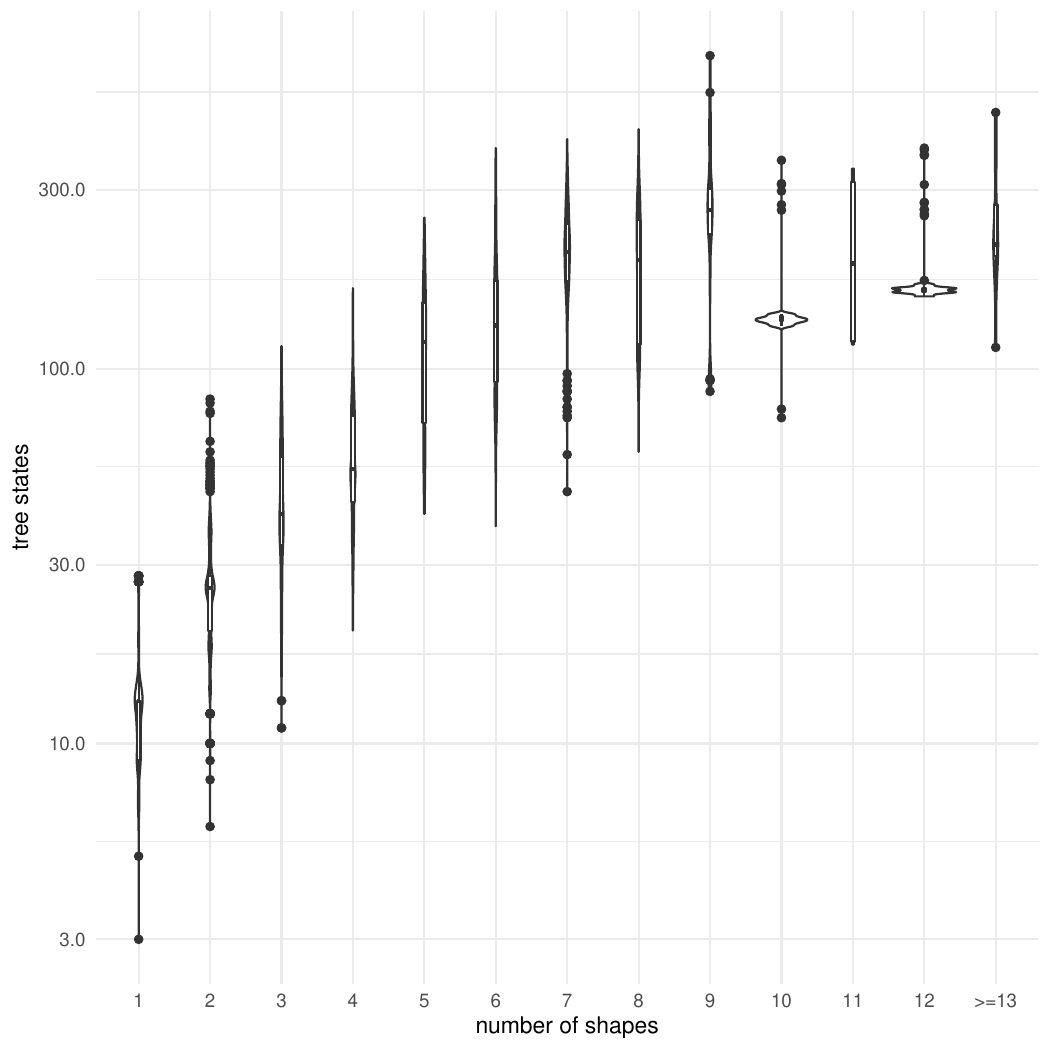}
   \includegraphics[width=.37\linewidth]{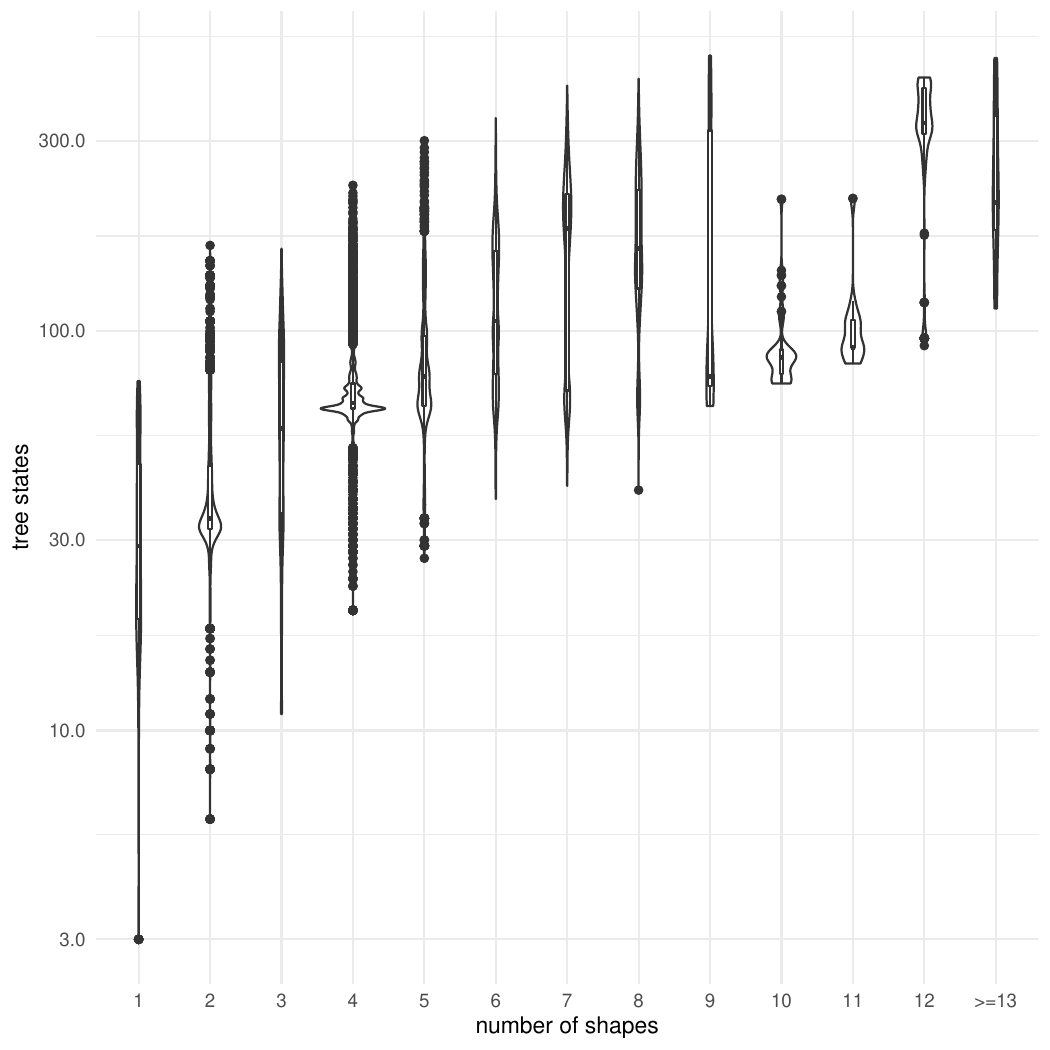}
\end{center}
\caption{Number of tree states}
\end{subfigure}
\caption{
The comparison of time, memory and tree states, in division by geometry 
(left Euclidean, right hyperbolic) and the number of shapes.
\label{correct_shps}}
\end{figure}

\begin{figure}[h]
\begin{center}
  \includegraphics[width=.32\linewidth]{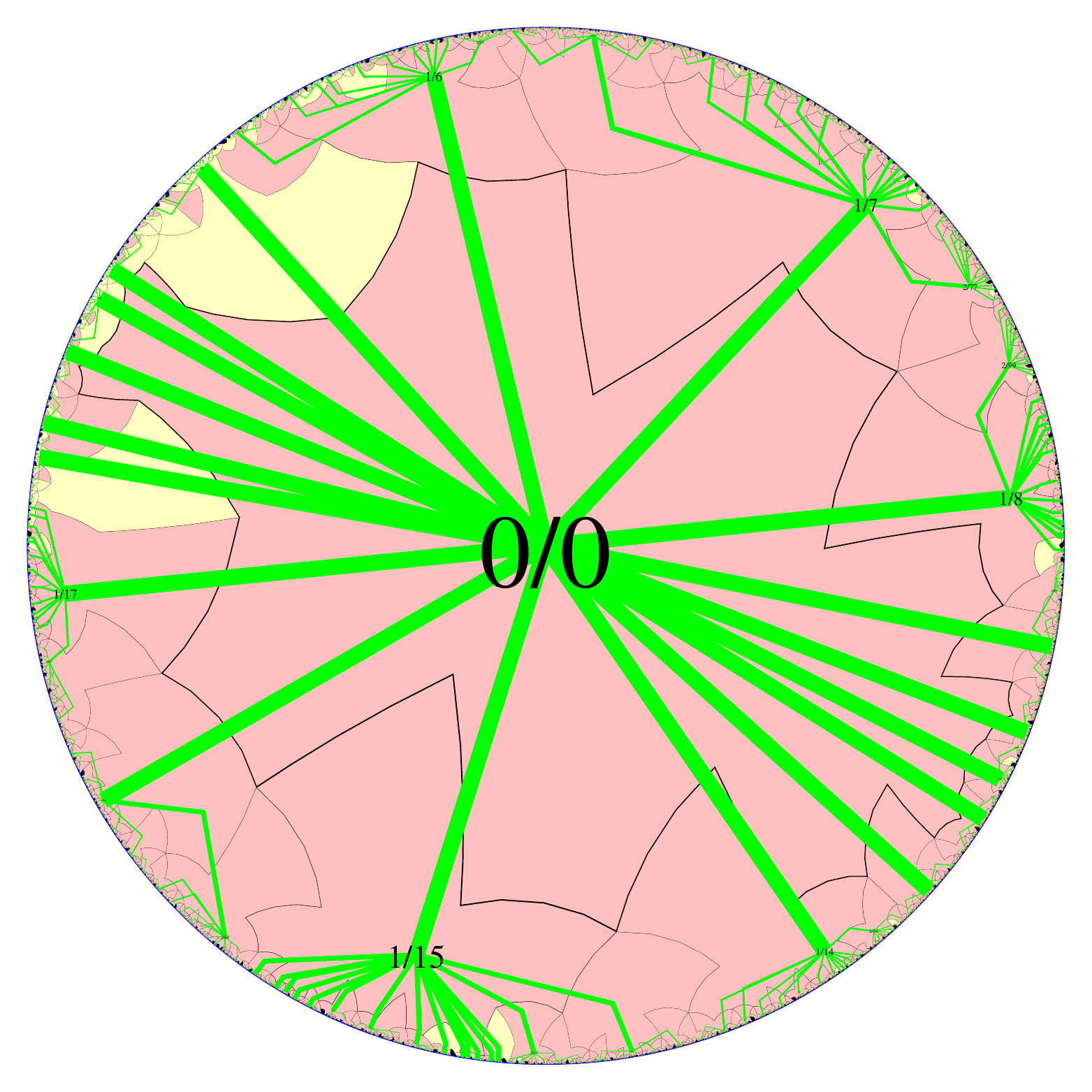}
  \includegraphics[width=.32\linewidth]{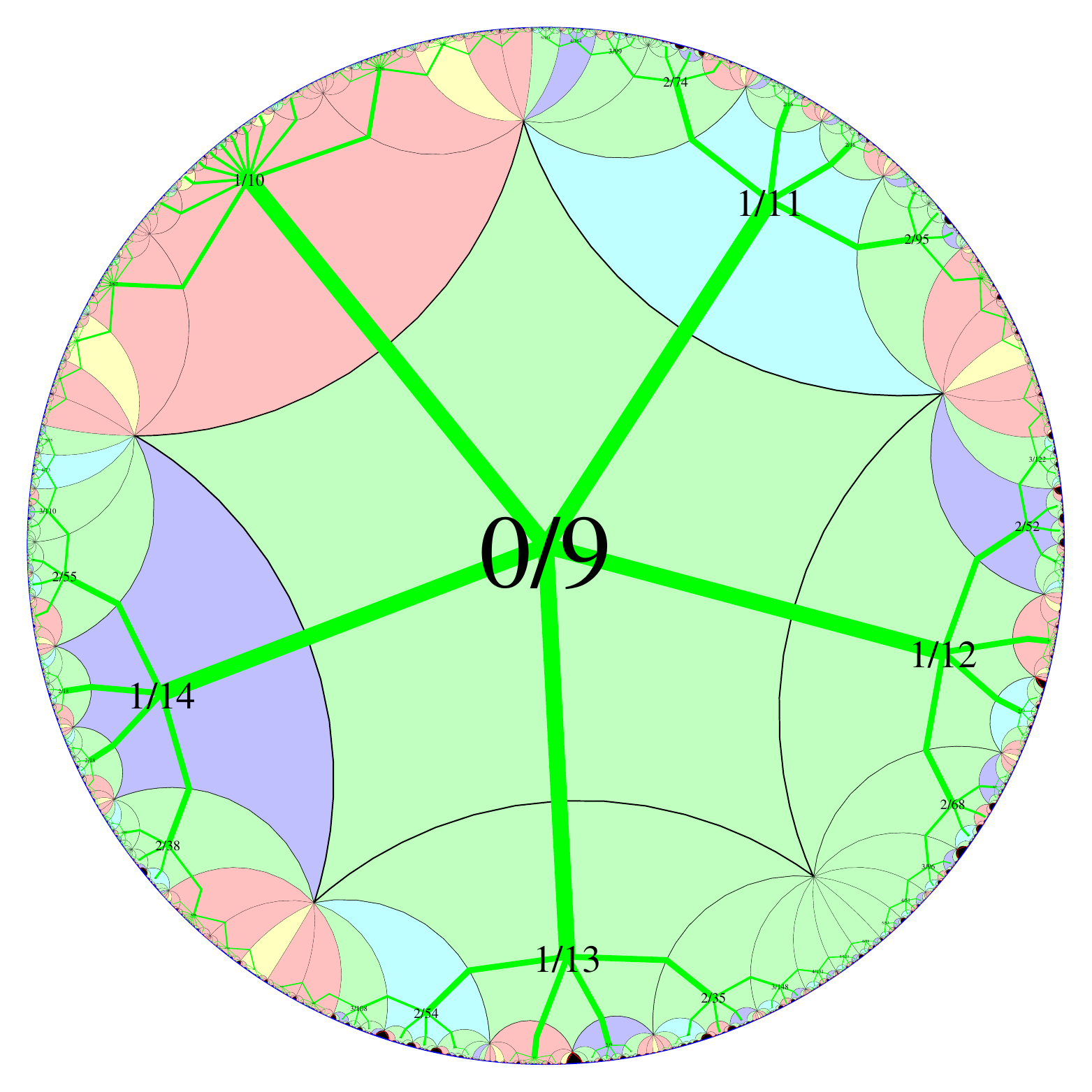}
  \includegraphics[width=.32\linewidth]{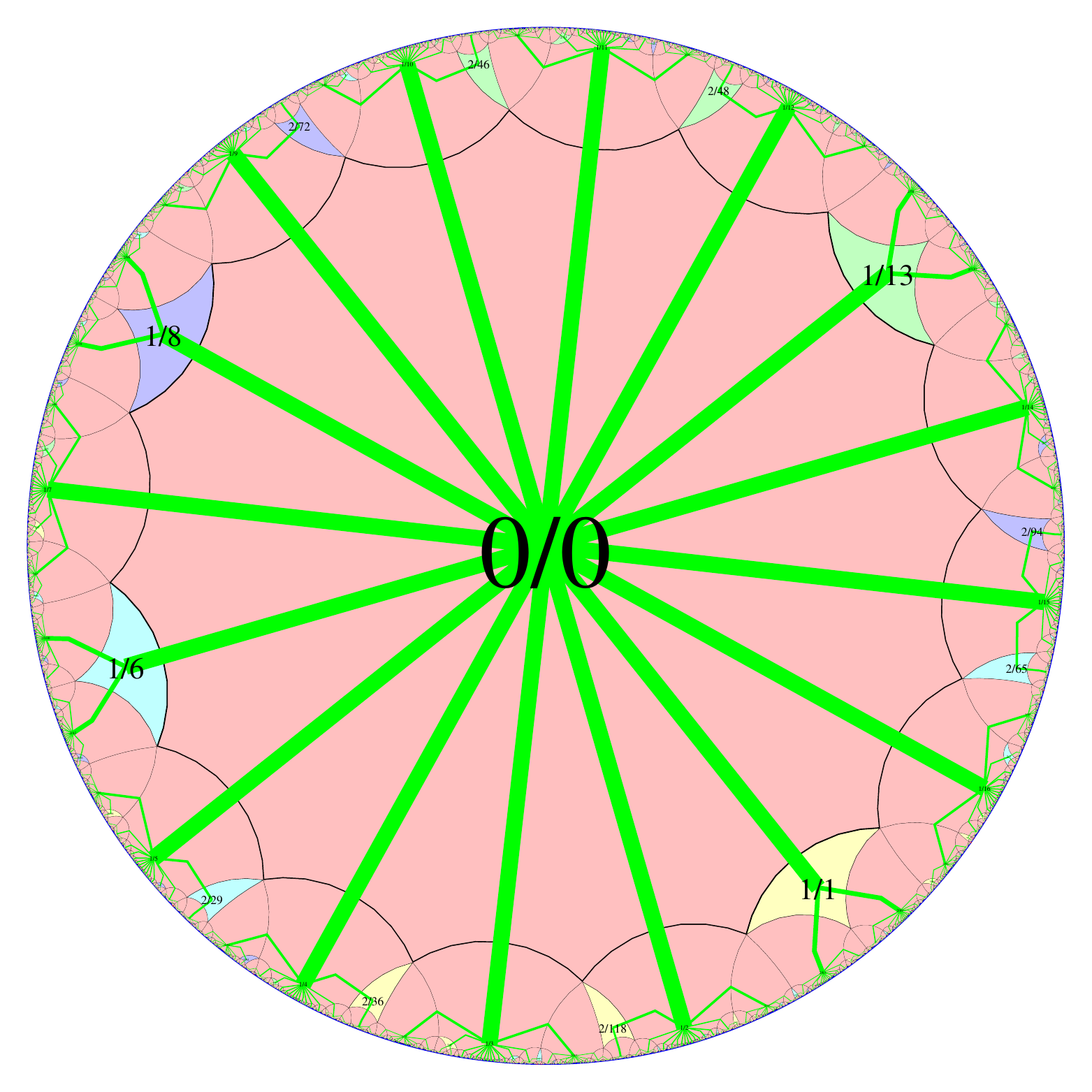}
  \includegraphics[width=.32\linewidth]{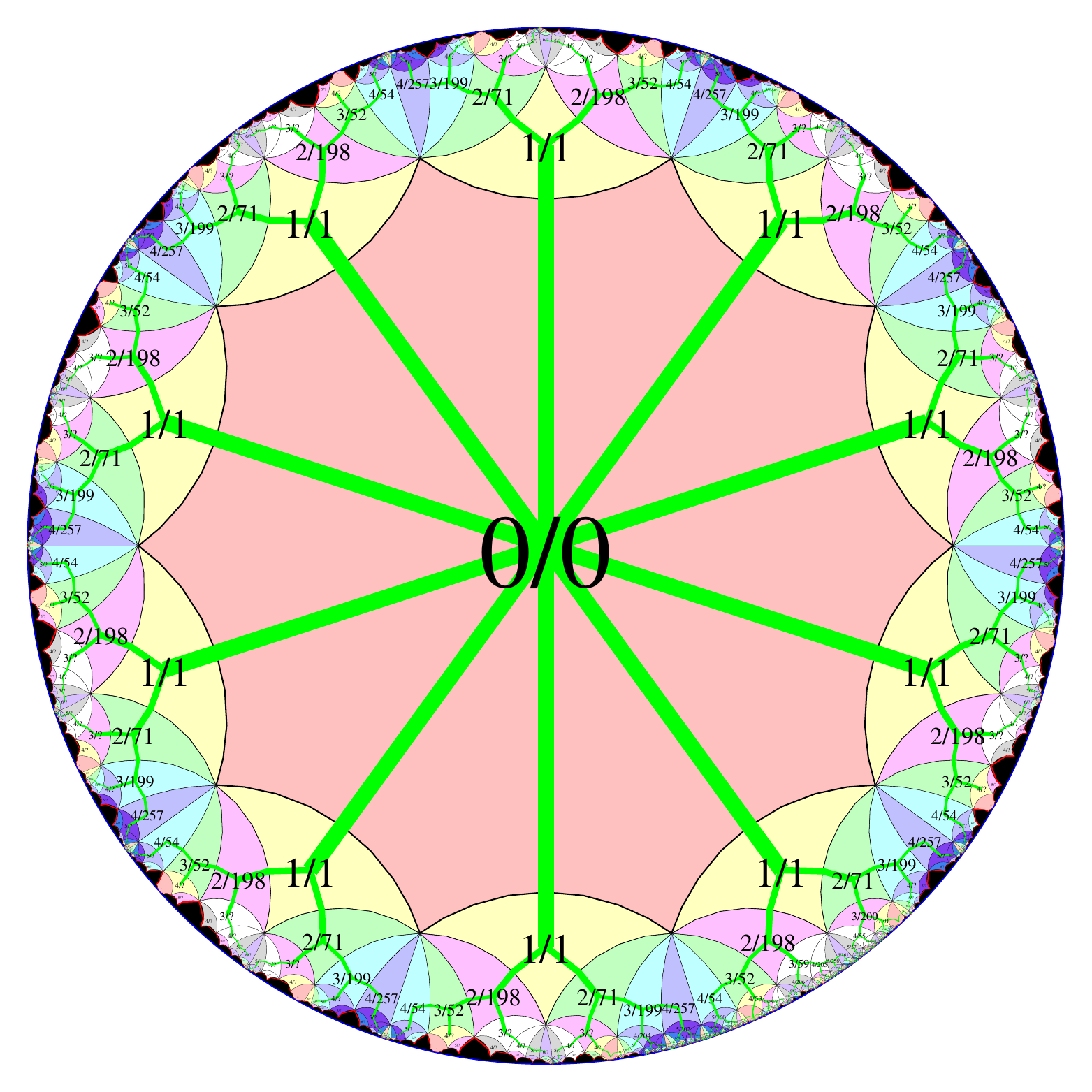}
  \includegraphics[width=.32\linewidth]{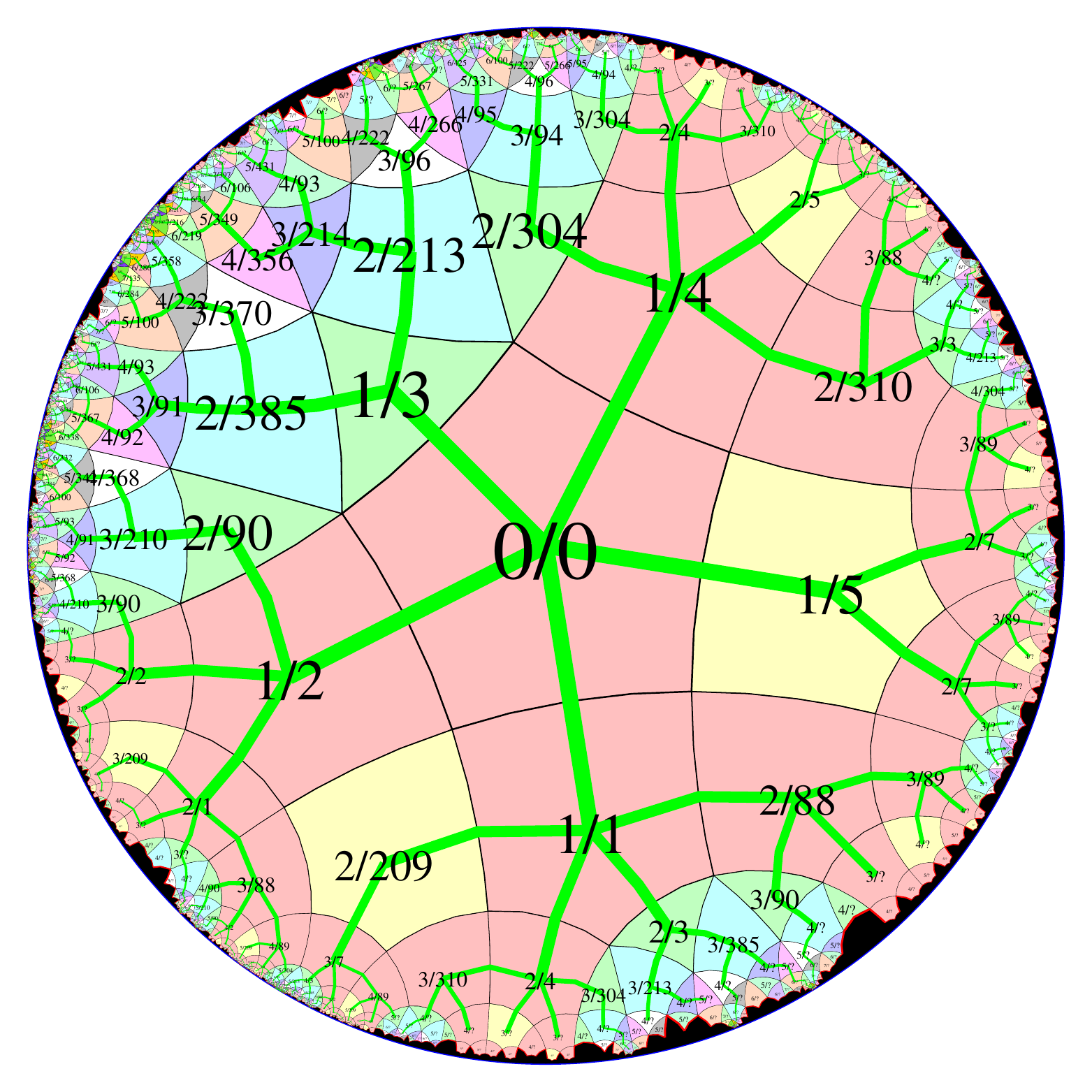}
  \includegraphics[width=.32\linewidth]{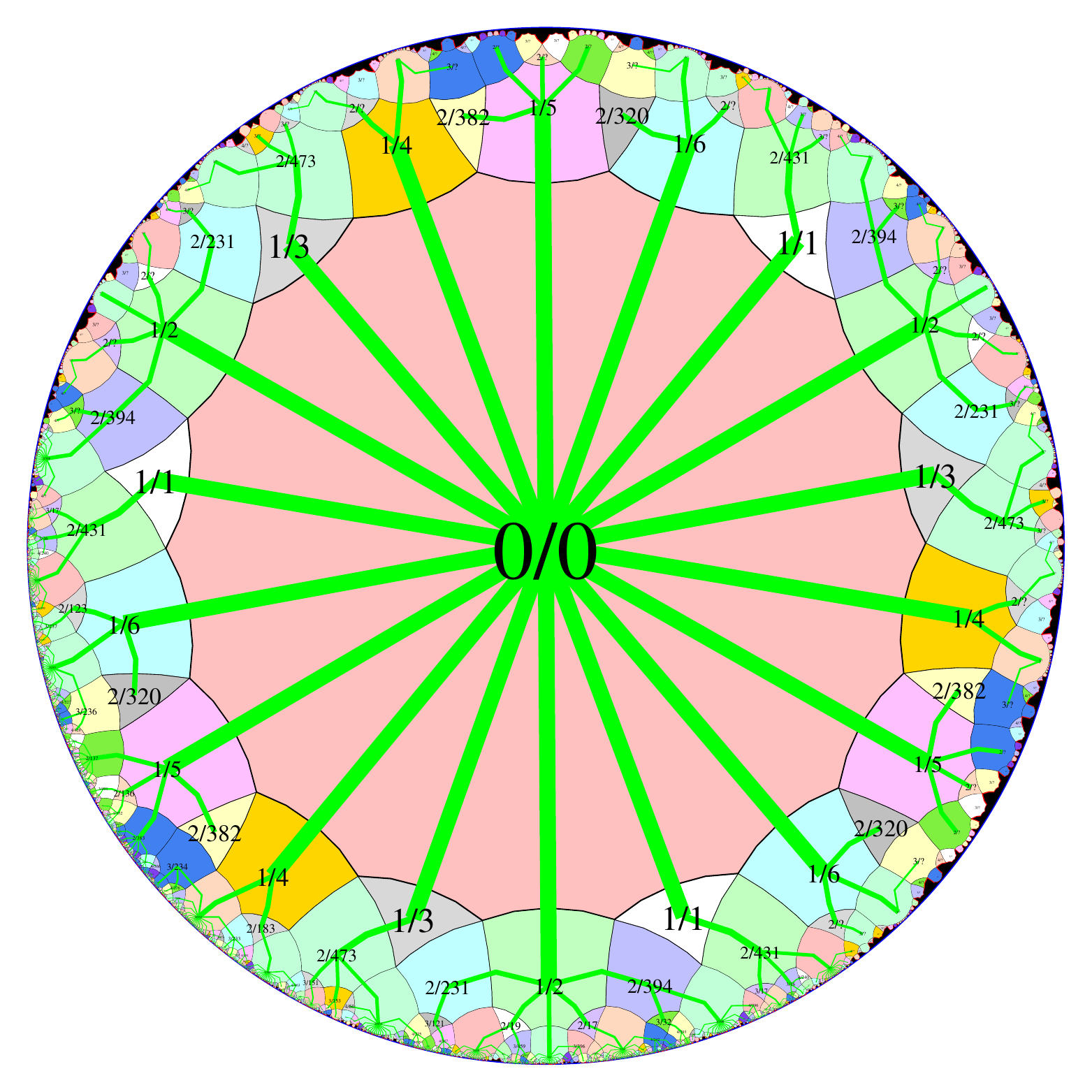}
\end{center}
\caption{Hyperbolic tessellations which achieve the highest time usage (in seconds), highest memory usage (k = thousands of cells), and the largest number of tree states. The red boundary shows which cells have been generated during the run of the algorithm.
Each tile gives distance from the center / tree state (if computed by the algorithm).
Top row: 7.1s, 2701k, 138 states; 2.8s, 1676k, 255 states; 1.1s, 902k, 129 states.
Bottom row: 0.6s, 143k, 625 states; 0.5s, 39k, 534 states; 0.3s, 193k, 501 states.}
\label{worstcase:h}
\end{figure}

\begin{figure}[h]
\begin{center}
  \includegraphics[width=.49\linewidth]{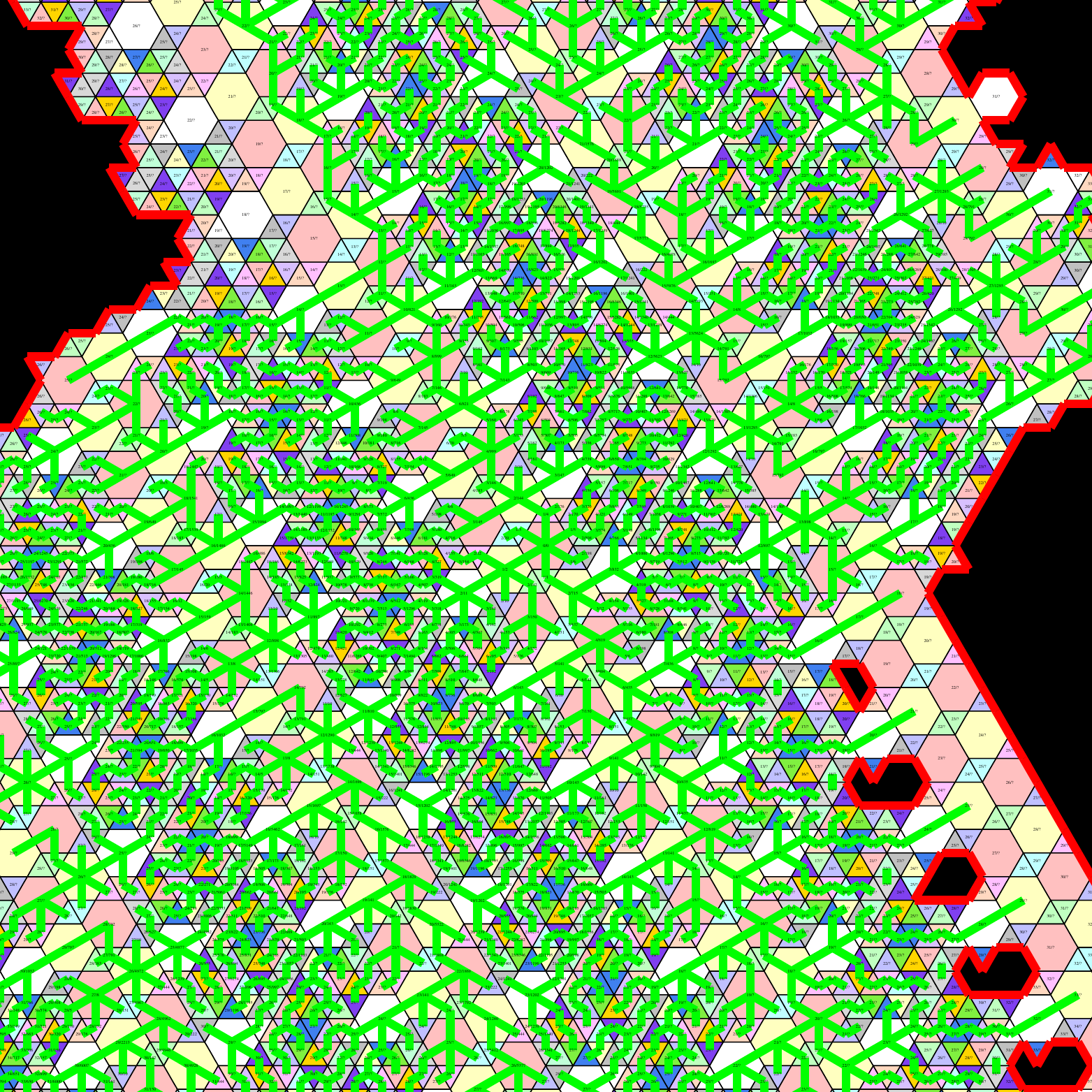}
  \includegraphics[width=.49\linewidth]{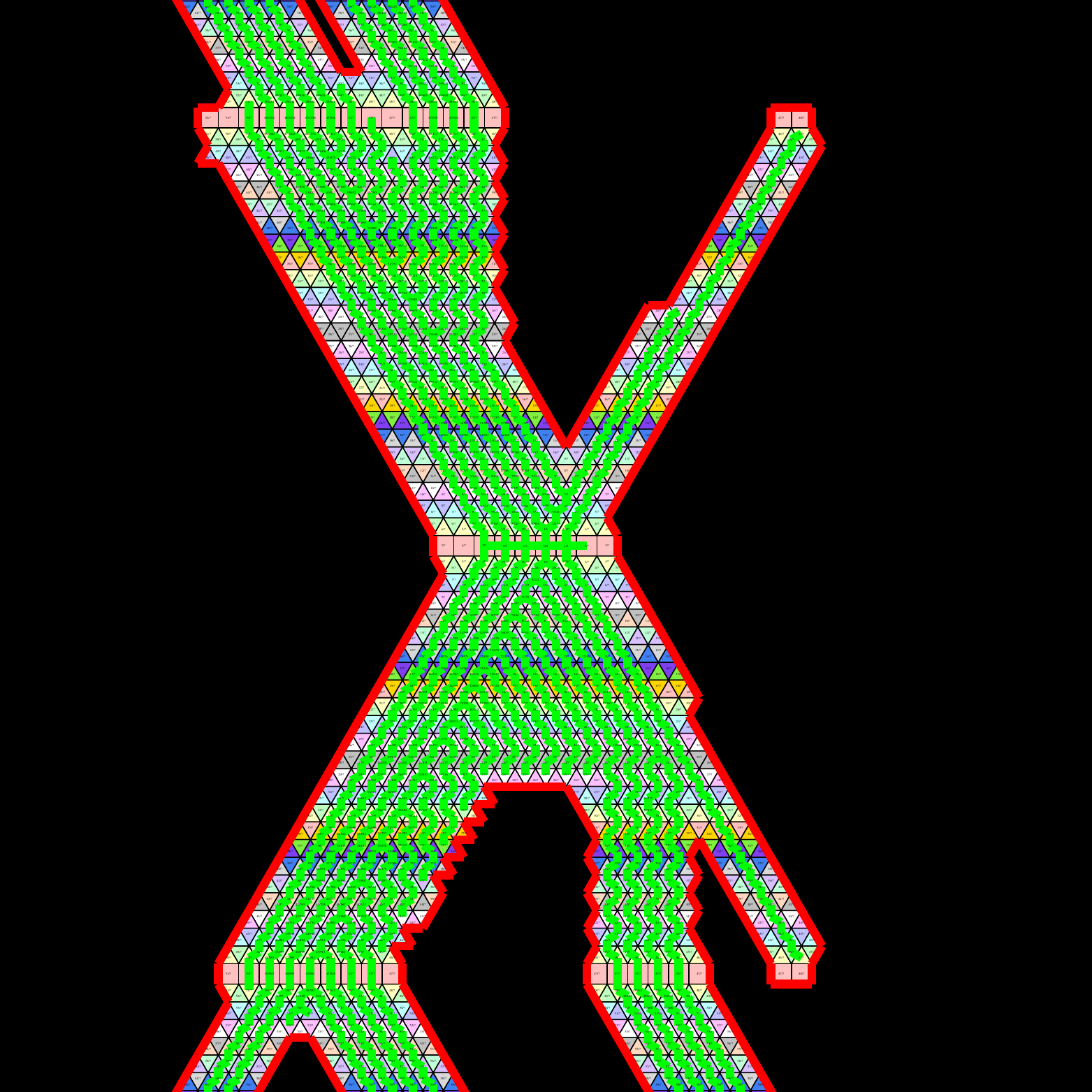}
  \includegraphics[width=.49\linewidth]{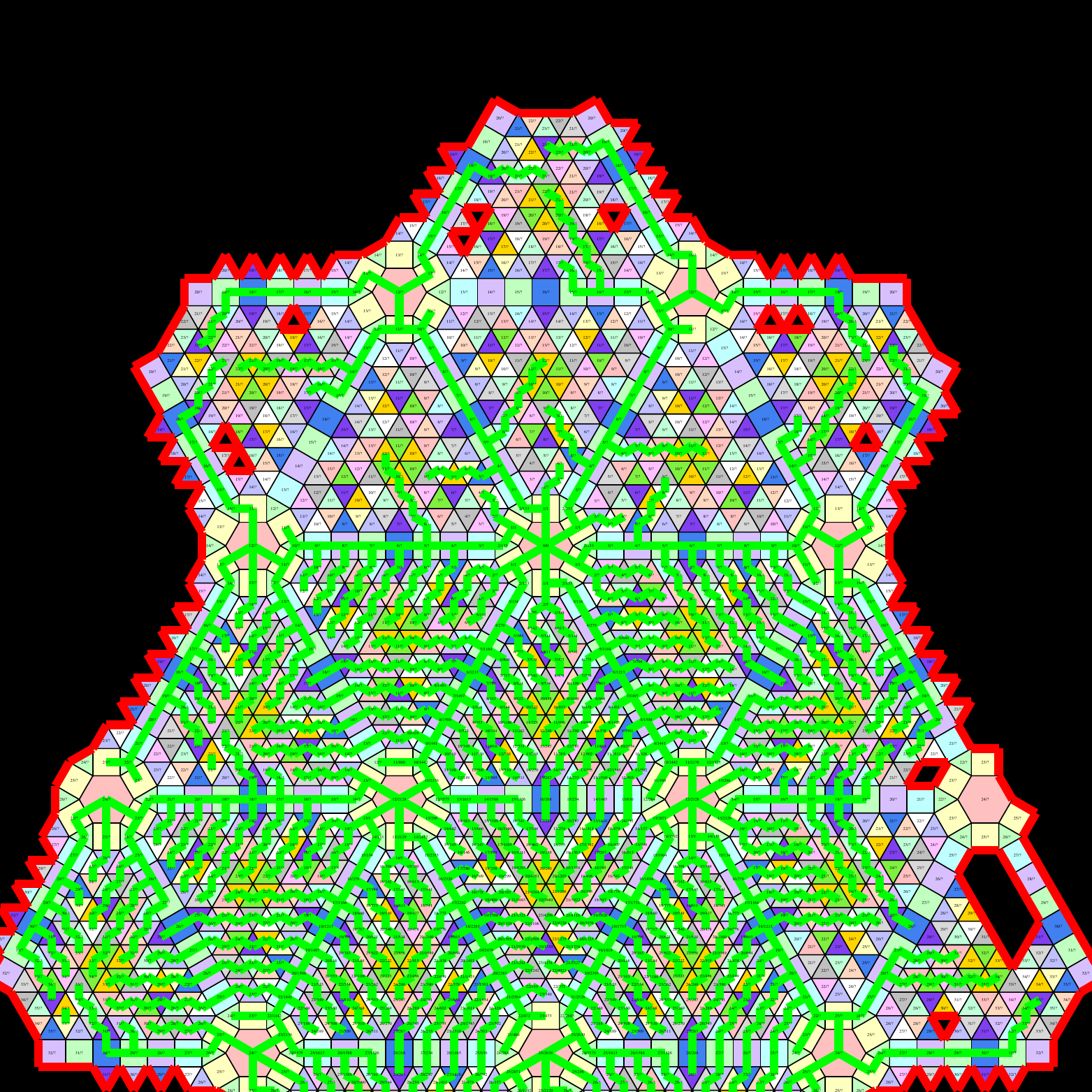}
  \includegraphics[width=.49\linewidth]{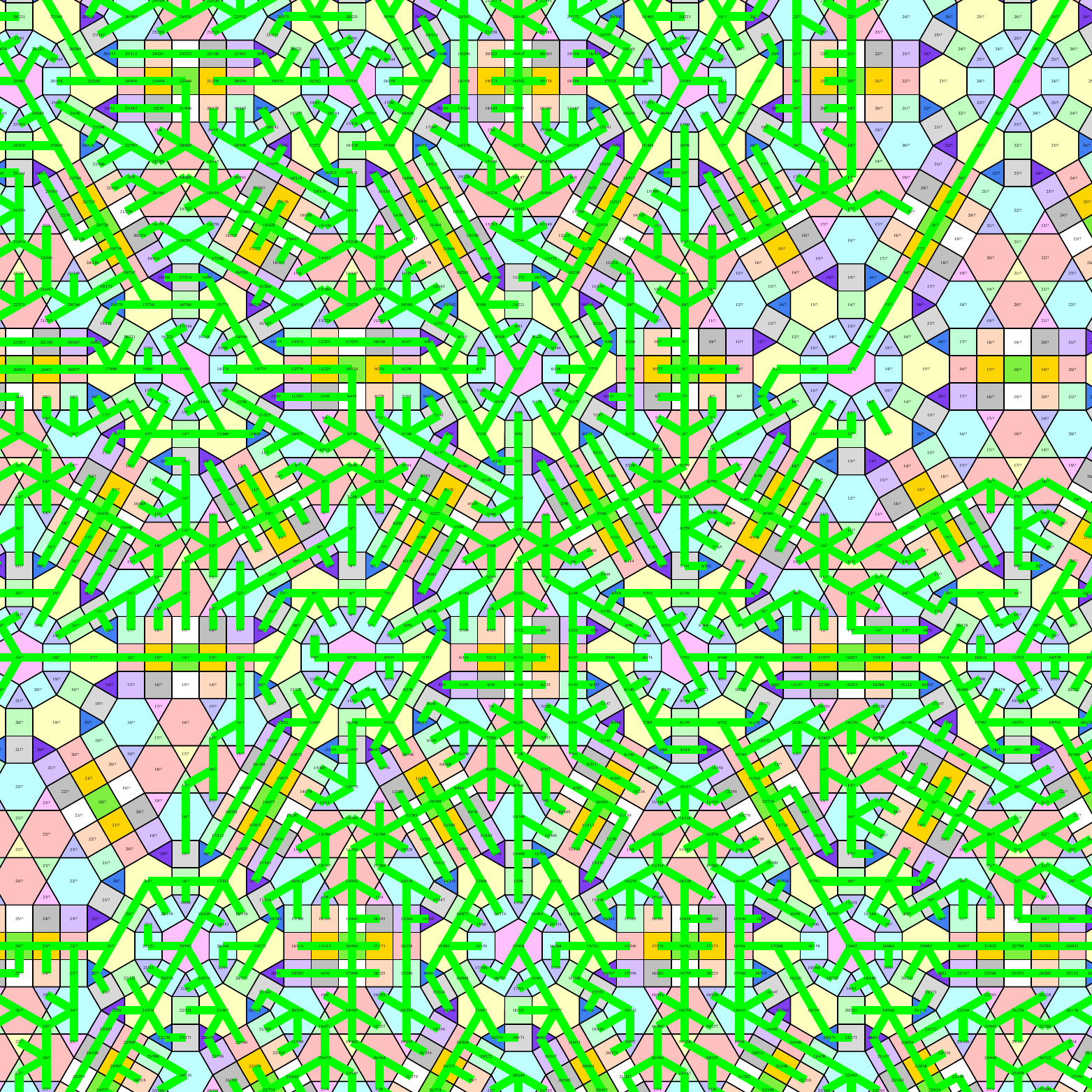}
\end{center}
\caption{Euclidean tessellations which achieve the highest time usage (in seconds), highest memory usage (k = thousands of cells), and the largest number of tree states. The red boundary shows which cells have been generated during the run of the algorithm.
Each tile gives distance from the center / tree state (if computed by the algorithm).
Top row: 257.7s, 1039k, 2538 states; 22.7s, 191k, 7297 states. 
Bottom row: 19.5s, 121k, 4440 states; 31.2s, 1187k, 1021 states.}
\label{worstcase:e}
\end{figure}


\begin{thebibliography}{10}

\bibitem{Bieberbach1911}
Ludwig Bieberbach.
\newblock {\"U}ber die bewegungsgruppen der euklidischen r{\"a}ume.
\newblock {\em Mathematische Annalen}, 70(3):297--336, Sep 1911.
\newblock \href {https://doi.org/10.1007/BF01564500}
  {\path{doi:10.1007/BF01564500}}.

\bibitem{Bieberbach1912}
Ludwig Bieberbach.
\newblock {\"U}ber die bewegungsgruppen der euklidischen r{\"a}ume (zweite
  abhandlung.) die gruppen mit einem endlichen fundamentalbereich.
\newblock {\em Mathematische Annalen}, 72(3):400--412, Sep 1912.
\newblock \href {https://doi.org/10.1007/BF01456724}
  {\path{doi:10.1007/BF01456724}}.

\bibitem{bogu_internet}
Marián Boguñá, Fragkiskos Papadopoulos, and Dmitri Krioukov.
\newblock Sustaining the internet with hyperbolic mapping.
\newblock {\em Nature Communications}, 1(6):1–8, Sep 2010.
\newblock URL: \url{http://dx.doi.org/10.1038/ncomms1063}, \href
  {https://doi.org/10.1038/ncomms1063} {\path{doi:10.1038/ncomms1063}}.

\bibitem{cannon}
James~W. Cannon, William~J. Floyd, Richard Kenyon, Walter, and R.~Parry.
\newblock Hyperbolic geometry.
\newblock In {\em In Flavors of geometry}, pages 59--115. University Press,
  1997.
\newblock Available online at
  \url{http://www.msri.org/communications/books/Book31/files/cannon.pdf}.

\bibitem{hrviz}
Dorota Celi\'nska and Eryk Kopczy\'nski.
\newblock Programming languages in github: {A} visualization in hyperbolic
  plane.
\newblock In {\em \ext{Proceedings of the Eleventh International Conference on
  Web and Social Media,} {ICWSM}, Montr{\'{e}}al, \ext{Qu{\'{e}}bec,} Canada,
  May 15-18, 2017.}, pages 727--728, Palo Alto, California, 2017. The AAAI
  Press.
\newblock URL:
  \url{https://aaai.org/ocs/index.php/ICWSM/ICWSM17/paper/view/15583}.

\bibitem{coxe41}
H.~S.~M. Coxeter.
\newblock The non-{E}uclidean symmetry of {E}scher's picture {C}ircle {L}imit
  {I}{I}{I}.
\newblock {\em Leonardo}, 12:19--25, 1979.

\bibitem{wpigroups}
David B.~A. Epstein, M.~S. Paterson, J.~W. Cannon, D.~F. Holt, S.~V. Levy, and
  W.~P. Thurston.
\newblock {\em Word Processing in Groups}.
\newblock A. K. Peters, Ltd., USA, 1992.

\bibitem{useofknuth}
D.B.A. Epstein, D.F. Holt, and S.E. Rees.
\newblock The use of knuth-bendix methods to solve the wordproblem in automatic
  groups.
\newblock {\em Journal of Symbolic Computation}, 12(4):397--414, 1991.
\newblock URL:
  \url{https://www.sciencedirect.com/science/article/pii/S0747717108800934},
  \href {https://doi.org/https://doi.org/10.1016/S0747-7171(08)80093-4}
  {\path{doi:https://doi.org/10.1016/S0747-7171(08)80093-4}}.

\bibitem{achilles}
William~J. Floyd, Brian Weber, and Jeffrey~R. Weeks.
\newblock The achilles' heel of o(3, 1)?
\newblock {\em Exp. Math.}, 11(1):91--97, 2002.
\newblock \href {https://doi.org/10.1080/10586458.2002.10504472}
  {\path{doi:10.1080/10586458.2002.10504472}}.

\bibitem{colorbook}
C.~Goodman-Strauss and N.~J.~A. Sloane.
\newblock A coloring-book approach to finding coordination sequences.
\newblock {\em Acta Crystallographica Section A}, 75(1):121--134, 2019.
\newblock URL:
  \url{https://onlinelibrary.wiley.com/doi/abs/10.1107/S2053273318014481},
  \href
  {http://arxiv.org/abs/https://onlinelibrary.wiley.com/doi/pdf/10.1107/S2053273318014481}
  {\path{arXiv:https://onlinelibrary.wiley.com/doi/pdf/10.1107/S2053273318014481}},
  \href {https://doi.org/https://doi.org/10.1107/S2053273318014481}
  {\path{doi:https://doi.org/10.1107/S2053273318014481}}.

\bibitem{hamilton2021nogo}
Linus Hamilton and Ankur Moitra.
\newblock No-go theorem for acceleration in the hyperbolic plane, 2021.
\newblock \href {http://arxiv.org/abs/2101.05657} {\path{arXiv:2101.05657}}.

\bibitem{trigrid}
Eryk Kopczynski and Dorota Celinska.
\newblock Hyperbolic grids and discrete random graphs.
\newblock {\em CoRR}, abs/1707.01124, 2017.
\newblock URL: \url{http://arxiv.org/abs/1707.01124}, \href
  {http://arxiv.org/abs/1707.01124} {\path{arXiv:1707.01124}}.

\bibitem{hyperrogue}
Eryk Kopczy\'{n}ski, Dorota Celi\'{n}ska, and Marek \v{C}trn\'{a}ct.
\newblock Hyper{R}ogue: Playing with hyperbolic geometry.
\newblock In {\em Proceedings of Bridges \ext{: Mathematics, Art, Music,
  Architecture, Education, Culture}}, pages 9--16, Phoenix, Arizona, 2017.
  Tessellations Publishing.

\bibitem{rtviz}
Eryk Kopczy\'nski and Dorota Celi\'nska-Kopczy\'nska.
\newblock Real-time visualization in non-isotropic geometries, 2020.
\newblock \href {http://arxiv.org/abs/2002.09533} {\path{arXiv:2002.09533}}.

\bibitem{parikhlics}
Eryk Kopczynski and Anthony~Widjaja To.
\newblock Parikh images of grammars: Complexity and applications.
\newblock In {\em Proceedings of the 2010 25th Annual IEEE Symposium on Logic
  in Computer Science}, LICS '10, pages 80--89, Washington, DC, USA, 2010. IEEE
  Computer Society.
\newblock URL: \url{http://dx.doi.org/10.1109/LICS.2010.21}, \href
  {https://doi.org/http://dx.doi.org/10.1109/LICS.2010.21}
  {\path{doi:http://dx.doi.org/10.1109/LICS.2010.21}}.

\bibitem{lampingrao}
John Lamping, Ramana Rao, and Peter Pirolli.
\newblock A focus+context technique based on hyperbolic geometry for
  visualizing large hierarchies.
\newblock In {\em Proceedings of the SIGCHI Conference on Human Factors in
  Computing Systems}, CHI '95, pages 401--408, New York, NY, USA, 1995. ACM
  Press/Addison-Wesley Publishing Co.
\newblock URL: \url{http://dx.doi.org/10.1145/223904.223956}, \href
  {https://doi.org/10.1145/223904.223956} {\path{doi:10.1145/223904.223956}}.

\bibitem{levygen}
Silvio Levy.
\newblock Automatic generation of hyperbolic tilings.
\newblock {\em Leonardo}, 25(3/4):349--354, 1992.
\newblock URL: \url{http://www.jstor.org/stable/1575861}.

\bibitem{margenstern_pentagrid}
Maurice Margenstern.
\newblock New tools for cellular automata in the hyperbolic plane.
\newblock {\em Journal of Universal Computer Science}, 6(12):1226--1252, dec
  2000.
\newblock \url|http://www.jucs.org/jucs\_6\_12/new\_tools\_for\_cellular|.

\bibitem{margenstern_heptagrid}
Maurice Margenstern.
\newblock Pentagrid and heptagrid: the fibonacci technique and group theory.
\newblock {\em Journal of Automata, Languages and Combinatorics},
  19(1-4):201--212, 2014.
\newblock \href {https://doi.org/10.25596/jalc-2014-201}
  {\path{doi:10.25596/jalc-2014-201}}.

\bibitem{munzner}
Tamara Munzner.
\newblock Exploring large graphs in 3d hyperbolic space.
\newblock {\em {IEEE} Computer Graphics and Applications}, 18(4):18--23, 1998.
\newblock URL: \url{http://dx.doi.org/10.1109/38.689657}, \href
  {https://doi.org/10.1109/38.689657} {\path{doi:10.1109/38.689657}}.

\bibitem{papa}
Fragkiskos Papadopoulos, Maksim Kitsak, M.~Angeles Serrano, Marian Bogu\~n\'a,
  and Dmitri Krioukov.
\newblock {Popularity versus Similarity in Growing Networks}.
\newblock {\em Nature}, 489:537--540, Sep 2012.

\bibitem{parikh}
Rohit Parikh.
\newblock On context-free languages.
\newblock {\em J. ACM}, 13(4):570--581, 1966.

\bibitem{tescatalog}
Marek \v{C}trn\'act and Eryk Kopczy\'nski.
\newblock {Tessellation Catalog}, 9 2021.
\newblock URL: \url{https://github.com/zenorogue/tes-catalog/}.

\end{thebibliography}
\end{document}